%% file: branching.tex
\documentclass[10pt]{article}

\usepackage{amsmath, amssymb, amsthm, amsfonts, graphicx}
\usepackage{cite}
\usepackage{geometry}

\usepackage{color}

\definecolor{orange}{rgb}{1,0.5,0}

\input{def-branching}

\numberwithin{equation}{section}
\begin{document}

\title{Domain structure of bulk ferromagnetic crystals in applied
 fields near saturation}

\author{Hans Kn\"upfer \thanks{Courant Institute of Mathematical
   Sciences, New York University, New York, NY 10012} \and Cyrill
 B. Muratov \thanks {Department of Mathematical Sciences, New Jersey
   Institute of Technology, Newark, NJ 07102}}
\maketitle
\begin{abstract}
 We investigate the ground state of a uniaxial ferromagnetic plate
 with perpendicular easy axis and subject to an applied magnetic
 field normal to the plate. Our interest is the asymptotic behavior
 of the energy in macroscopically large samples near the saturation
 field. We establish the scaling of the critical value of the applied
 field strength below saturation at which the ground state changes
 from the uniform to a branched domain magnetization pattern and the
 leading order scaling behavior of the minimal energy. Furthermore,
 we derive a reduced sharp-interface energy giving the precise
 asymptotic behavior of the minimal energy in macroscopically large
 plates under a physically reasonable assumption of small deviations
 of the magnetization from the easy axis away from domain walls. On
 the basis of the reduced energy, and by a formal asymptotic analysis
 near the transition, we derive the precise asymptotic values of the
 critical field strength at which non-trivial minimizers (either
 local or global) emerge. The non-trivial minimal energy scaling is
 achieved by magnetization patterns consisting of long slender
 needle-like domains of magnetization opposing the applied field.
\end{abstract}

% {\it AMS classification: } Primary: 82D40, Secondary: 49S05, 35A15.

% {\it Keywords: } domain pattern, branching, micromagnetics.

%\setcounter{tocdepth}{1} % only sections
\tableofcontents

\section{Introduction}

Ferromagnetic materials offer a fascinating example of physical
systems capable of producing an extraordinarily rich variety of
spatial patterns \cite{hubert}. By a pattern in a ferromagnet, one
usually understands a stable spatial distribution of the magnetization
vector in the sample. This definition reflects the mesoscopic nature
of the magnetization patterns: they are observed on the length scales
significantly exceeding the atomic scale (making the definition of the
magnetization per unit volume meaningful), yet they are susceptible to
small random fluctuations due to thermal noise, with the noise
providing a selection mechanism for observable patterns.

\medskip

On the mesoscopic level, the theory describing the spatio-temporal
dynamics of the magnetization patterns in ferromagnetic materials is
formulated in terms of partial differential equations (with a possible
addition of stochastic forcing \cite{brown63}) for the magnetization
vector $\mathbf M = \mathbf M(\mathbf r, t)$
\cite{landau8,landau9,hubert}.  At the center of the theory is the
micromagnetic energy functional $\EEE [\mathbf M]$ describing the
contributions of different physical interactions (for specifics, see
the following section) \cite{hubert,desimone00,kohn07iciam}.
Magnetization patterns are viewed as global or, more generally, local
minimizers of $\EEE$, forming mainly due to the competition of the
exchange, anisotropy, and the magnetostatic interactions, with the
applied external field playing a significant role
\cite{hubert,landau8}. Because of the non-local nature of the
magnetostatic forces, their effect can depend significantly on the
geometry of the ferromagnetic sample
\cite{desimone02,desimone00,hubert,kohn07iciam}.

\medskip

In bulk crystalline materials the local anisotropy energy and the
short-ranged exchange energy act jointly to favor magnetization
distributions in the form of extended {\em magnetic domains} in which
the magnetization vector stays nearly constant, separated by {\em
  domain walls}, where the magnetization direction changes
abruptly. It was already realized in the pioneering works of Landau
and Lifshitz \cite{landau35} and Kittel \cite{kittel46} that, while
the structure of the domain walls may not be significantly affected by
the long-range magnetostatic forces, these forces should determine the
relative spatial arrangement of the domains with different orientation
of the magnetization. In fact, since the total magnetostatic energy
scales faster than volume as the size of the system increases, in
large samples the effect of long-range magnetostatic interactions
becomes dominant. As a result, the magnetization patterns develop
rapid oscillations to cancel out the induced magnetic field and form
intricate structures, which are generally referred to as {\em branched
  domains}, even though the actual topological branching of the
domains is not really required.

\medskip

Despite a long history of observations of branched domain structures
in ferromagnetic materials \cite{hubert} and related systems (see e.g.
\cite{strukov,landau8,prozorov05,prozorov07,shur00}), mathematical
understanding of the branching phenomenon started to emerge only
recently with the ansatz-free analysis of energy minimizing structures
\cite{choksi99, choksi08} (there is, of course, an extensive
literature of ansatz-based studies, see e.g.
\cite{privorotskii72,gabay85,kaczer64}). In particular, for bulk
crystalline ferromagnets in the absence of an applied field the first
rigorous analysis of the branched domain structures was performed in
the work of Choksi and Kohn \cite{choksi98}. They studied a sharp
interface version of the micromagnetic energy and were able to obtain
matching (in the sense of scaling with the sample thickness) upper and
lower bounds for the energy of minimizers of the reduced energy. We
note that the connection of the sharp interface energy to the full
micromagnetic energy in the limit of high anisotropy was recently
established in \cite{OttoViehmann-2008}. The results of
\cite{choksi99} are suggestive that the energy minimizers of the sharp
interface micromagnetic energy are in some sense not very different
from the branched domain ansatz used as a trial function in the
calculation of the upper bound of the energy of the minimizers. The
latter shares many common features with the branched domain structures
observed in experiments \cite{hubert}.  Since then, similar results
have also been obtained for models describing type-I superconductors
in the intermediate state \cite{choksi08,choksi04} and
diblock-copolymers undergoing microphase separation
\cite{choksi01,choksi09,m:cmp09}.

\medskip

Note that the presence of a moderate applied magnetic field does not
alter the situation qualitatively. On the other hand, if a very strong
external magnetic field is applied to the sample, then it will
obviously overwhelm all other effects and result in a uniform
magnetization pattern in the direction of the applied field.  It is
then clear that a bifurcation from the uniform to a non-uniform
magnetization pattern will occur when the field strength is gradually
reduced. Let us point out that this transition would typically occur
via nucleation and growth of new domains and is, therefore,
accompanied by a hysteresis. In other words, in a certain range of
applied fields one should find coexistence of different types of
patterns. Their relative stability and the transition pathways between
them are, therefore, important questions to be addressed. Note that
these questions also naturally arise in various other problems of
energy driven pattern formation, such as type-I and type-II
superconductors and Ginzburg-Landau models with Coulomb repulsion
\cite{choksi04,choksi08,aftalion07,m:cmp09}.

\subsection*{Main results}

We investigate the properties of the magnetization patterns in bulk
uniaxial crystalline ferromagnets in the presence of external magnetic
field applied along the material's easy axis.  We are interested in
the transition to non-trivial energy minimizers occurring near the
saturation field in ferromagnetic plates with perpendicular easy axis.

\medskip

In Section \ref{sec-bulk}, we establish the scaling behavior of the
minimal energy in dependence of the plate thickness and the exterior
field. The precise result is stated in Theorems \ref{thm-bulk} and
\ref{thm-diffuse}. As a consequence, we get that in macroscopically
large plates the transition from the monodomain to the branched
magnetization pattern occurs when the strength $H_{\rm ext}$ of the
applied field $H_{\rm ext}$ satisfies
\begin{align} \label{eq:47} %
  H_s - H_\mathrm{ext} \ \sim \ \left\{ {A K^2 \over L^2 M_s (K + 4
      \pi M_s^2) } \ln \left( {4 \pi^2 L^2 M_s^4 \over A K} \right)
  \right\}^{\frac 13},
\end{align}
where $H_s = 4 \pi M_s$ is the saturation field. We refer to Section
\ref{sec-model} for the precise definitions of the physical parameters
in \eqref{eq:47}. For smaller applied fields the energy of the
minimizers per unit area in macroscopically large plates is
\begin{align} \label{eq:48} %
  {\mathrm{Energy} \over \mathrm{Area}} \ \sim \ \left\{ {L A K^2
      M_s^2 \over K + 4 \pi M_s^2} \left( 1 - {H_\mathrm{ext} \over
        H_s} \right)^3 \ln \left( 1 - {H_\mathrm{ext} \over H_s}
    \right)^{-1} \right\}^{\frac 13}.
\end{align}
In particular, the energy per unit area of the plate scales as
$L^{\frac 13}$ with the plate thickness and linearly (up to a slow
logarithmic dependence) with the deviation of the applied field from
the saturation field. This energy is achieved by trial functions
consisting of periodic patterns of slender needle-like disconnected
domains of magnetization opposing the applied field. In each unit cell
of such a trial function the magnetization pattern refines toward the
plate boundaries in a self-similar fashion (see
Fig. \ref{fig-needles}). This class of magnetization patterns is,
therefore, a natural candidate for the precise form of the energy
minimizers (see also Fig. \ref{cell3d}).

\begin{figure}
\hspace{3.2cm} {\bf a)} \\ \\
\centerline{\includegraphics[width=2cm,angle=90]{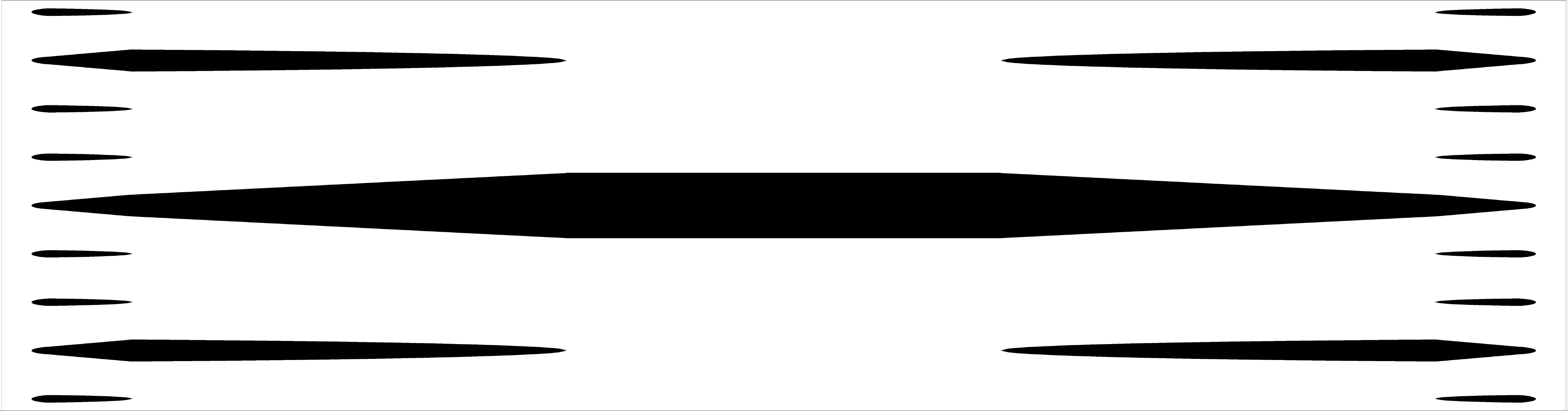}} \\

\hspace{3.2cm} {\bf b)} \\ 
\centerline{\includegraphics[width=7.5cm]{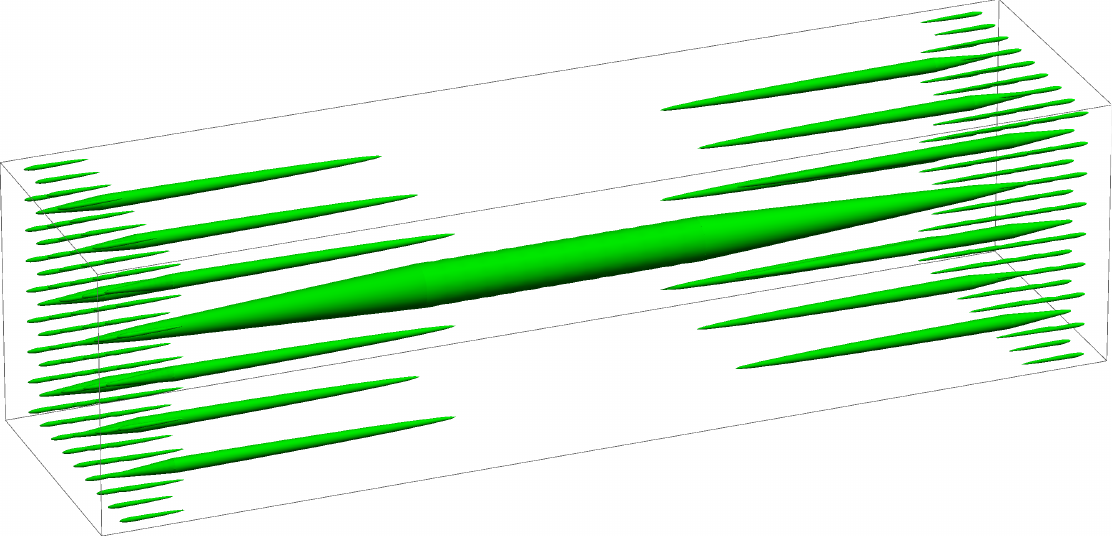}}
\caption{A sketch of the refining needle configuration. (a) The
  projection of one period of the domain pattern on the
  $x_1x_2$-plane. (b) A three-dimensional sketch of one period of the
  domain pattern. Shaded regions indicate the domains of magnetization
  opposing the applied field.}
\label{fig-needles}
\end{figure}

\medskip

In Section \ref{sec-reduced}, we further investigate the asymptotic
behavior of the energy in macroscopic samples. Under a physically
reasonable assumption that the magnetization vector does not deviate
strongly from the easy axis, we rigorously derive a reduced energy,
whose minimum agrees asymptotically with the sharp interface version
of the energy, see Theorem \ref{t:E0}. The obtained result assigns a
mathematical meaning to the $\mu^*$-method for computing the energy
contributions away from the domain walls in a magnetization pattern,
which was proposed more than half a century ago in the physics
literature \cite{williams49}.  The obtained reduced energy, given by
(\ref{eq:21}), practically coincides with that of an infinitely hard
material in which the strength of magnetostatic interaction has been
suitably renormalized. The latter explains why the behavior of minimal
energy in both hard and soft materials is the same up to a certain
factor in the macroscopic limit. Let us also note that for the same
reason the energy per unit area becomes essentially independent of the
saturation magnetization $M_s$ in soft materials with fixed value of
$H_\mathrm{ext} / H_s$, see \eqref{eq:48}.

\medskip

In Section \ref{sec-transition}, we perform a formal asymptotic
analysis of the reduced energy in (\ref{eq:21}) and establish a
precise asymptotic behavior of the critical field $H_{c_0}$ at which
the only minimizer (global or local) is expected to be the uniform
state, and the critical field $H_{c_1}$ at which non-trivial
minimizers emerge, see Theorem \ref{thm-needle}. It turns out that
asymptotically for macroscopically large plates
\begin{align} \label{eq:49} %
  1 - {H_{c_{0,1}} \over H_s} \ \simeq \ C_{0,1} \left\{ {A K^2 \over
      L^2 M_s^4 (K + 4 \pi M_s^2) } \ln \left( {4 \pi^2 L^2 M_s^4
        \over A K} \right) \right\}^{\frac 13},
\end{align}
where $C_0 \approx 0.4368$ and $C_1 \approx 0.5403$. At $H \sim
H_{c_{0,1}}$ the magnetization patterns are expected to consist of
slender, approximately radially-symmetric needle-like domains spanning
the entire plate thickness and separated by large distances compared
to the needle radius. Equation \eqref{eq:49} is obtained from a
reduced one-dimensional expression for the energy of needles, see
\eqref{eq:25}.  Solving the respective Euler-Lagrange equation
exactly, we obtain the precise shape of the needle and,
correspondingly, the expression in \eqref{eq:49}.

\subsection*{Structure of the paper and notations}

The paper is structured as follows: In Section \ref{sec-model}, we
present the micromagnetic energy functional and introduce its sharp
interface version. In Section \ref{sec-bulk}, we prove matching upper
and lower bounds for bulk samples near the critical field. In Section
\ref{sec-reduced}, we derive a reduced model that captures the leading
order energy in the macroscopic limit. In Section
\ref{sec-transition}, we perform a further reduction of the energy and
find the precise location of the transition to non-trivial minimizers
by solving the reduced minimization problem exactly.

\medskip

We will denote a generic point in space by $x = (x_1, x_2, x_3) =
(x_1, x_\perp)$, where $x_1$ is the component in the direction of the
easy axis and $x_\perp = (x_2, x_3)$ is the component projection onto
the plane normal to the easy axis.  Similarly, we will denote the
component of a vector $\mathbf v \in \mathbb R^3$ in the direction of
the easy axis by $v_1$ and its projection to the plane normal to the
easy axis by $\mathbf v_\perp = (v_2, v_3)$. The spatial gradient is
similarly separated into the components along and perpendicular to the
easy axis: $\nabla = (\partial_1, \nabla_\perp)$.

\medskip

We use the symbols $\sim, \lesssim$ and $\gtrsim$ to indicate that an
estimate holds up to a universal constant. For example $A \sim B$
means that there are universal constants $c, C > 0$ such that $cA \leq
B \leq CA$. The symbols $\ll$ and $\gg$ indicate that an estimate
requires a small universal constant. For example, if we say that $A
\lesssim B$ for $\eps \ll 1$, this is a short way of saying that $A
\leq CB$ holds for all $\eps \leq \eps_0$ where $\eps_0 > 0$ is a
small universal constant. By the symbol $\simeq$, we indicate
asymptotic equivalence of two expressions: E.g. by writing $A \simeq
B$ for $\eps \ll 1$ and $\lambda \ll 1$, we mean that for every
$\delta > 0$, there are $\eps_0, \lambda_0 > 0$ such that $|A/B - 1|
\leq \delta$ for all $\eps < \eps_0$ and all $\lambda < \lambda_0$.
% Universal constants are denoted by $C$.

\begin{figure}
%  \hspace{8ex}
\hspace{3cm} {\bf a)} \hspace{4cm} {\bf b)}
\vspace{-8mm}
\begin{center}
\includegraphics[width=3.35cm]{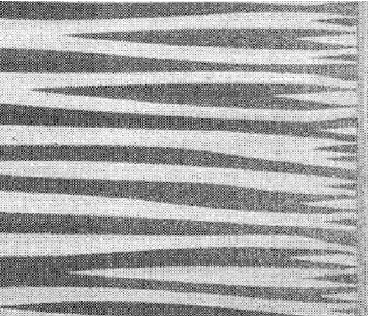}
\hspace{1cm}
\includegraphics[width=3.5cm]{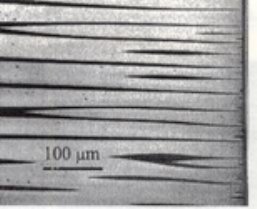}
\end{center}
\caption{Side views of magnetization domain patterns refining towards
  the boundary in bulk cobalt crystals: in the absence of the magnetic
  field (a) and in the applied field at 60\% to saturation (b). From
  Ref.  \cite{hubert}.}
\label{cell3d}
\end{figure}

\section{Physical model and sharp interface energy}
\label{sec-model}

{\it Micromagnetic energy:} Up to an additive constant, the
micromagnetic energy for a mono-crystalline uniaxial ferromagnet (see
e.g. \cite{hubert}, using CGS units) is given by
\begin{align}
 \label{eq:1}
 \begin{aligned}
   \EEE [\MM] &= \int_{\widetilde\Omega} \Bigl( {A \over 2 M_s^2}
   |\nabla \MM|^2 + {K \over 2 M_s^2} |\MM_\perp|^2 -
   \mathbf H_\mathrm{ext} \cdot \MM \Bigr) \, d^3 r \\
   &\qquad+ \frac12 \int_{\mathbb R^3} \int_{\mathbb R^3} {\nabla
     \cdot \MM(\mathbf r) \, \nabla \cdot \MM(\mathbf r') \over
     |\mathbf r - \mathbf r'|} \, d^3 r \, d^3 r' + {1 \over 8 \pi}
   \int_{\widetilde\Omega} |\mathbf H_\mathrm{ext}|^2 \, d^3 r
 \end{aligned}
\end{align}
Here, $\widetilde\Omega \subset \mathbb R^3$ describes the region of
space occupied by the ferromagnetic material, and the magnetization
vector $\MM : \R^3 \to \mathbb R^3$ satisfies $|\MM| = M_s$ in
$\widetilde\Omega$ and $\MM = 0$ outside.  The terms in the energy, as
they appear in the formula, are:
\begin{enumerate}
\item The {\em exchange} energy favoring a uniform magnetization.

\item The {\em anisotropy} energy favoring alignment of the
 magnetization with the easy axis.

\item The {\em Zeeman} energy favoring alignment with the external
 field $\mathbf H_\mathrm{ext}$.

\item The {\em stray field} energy describing long-range Coulomb
  interactions of the ``magnetic charges'' $\nabla \cdot \MM$.

\item A {\em constant} term, added for convenience.
\end{enumerate}
Also in (\ref{eq:1}), $A$ is the exchange constant, $M_s$ is the
saturation magnetization and $K$ is the uniaxial anisotropy
constant. The subscript ``$\perp$'' denotes the components of a vector
in the plane normal to the easy axis.

\medskip

{\it Geometry of the sample:} We consider a plate of constant
thickness $L$, whose surfaces are oriented in the direction normal to
the easy axis. To simplify the issues associated with the treatment of
the lateral boundaries of the sample, we assume periodicity with
period $\mathcal L$ in the plane normal to the easy axis.  We hence
write $\widetilde\Omega = (0, L) \times \widetilde\TTT$, where
$\widetilde\TTT = [0, \LLL )^2$ is a torus with periodicity $\LLL$.
The periodicity assumption, however, is not essential, as long as the
energy of the minimizers is extensive in $\mathcal L$, i.e. we have
$\inf \mathcal E = O(\mathcal L^2)$ as $\mathcal L \to \infty$.  As we
will show below, this will indeed be the case. Also, the external
field is assumed to be in the direction of the easy axis (and hence
also normal to the material surface):
\begin{align*}
 \mathbf H_\mathrm{ext} = H_\mathrm{ext} \, \mathbf e_1,
\end{align*}
where $\mathbf e_1$ is the unit vector in the direction of the easy
axis.

\medskip

% \begin{figure}[ht] %
%   \centering
%   \includegraphics[width=5cm]{geometry}   
%   \caption{Geometry}
% \end{figure}

{\it Rescaling:} As usual, we first introduce the exchange length
$l_\mathrm{ex} = \sqrt{A/(4 \pi M_s^2)}$ and the dimensionless quality
factor $Q = K/(4 \pi M_s^2)$ \cite{hubert}.  Introducing $\mm = \MM /
M_s $, $ \mathbf h_\mathrm{ext} = \mathbf H_\mathrm{ext}/( 4 \pi
M_s)$, $\ell = \mathcal L / L$, and measuring lengths and energy in
units of $L$ and $2 \pi M_s^2 L^3$, respectively, we can then rewrite
\eqref{eq:1} as
\begin{align}
\label{eq:8}
\EEE [\mm] = \int_\Omega {l_\mathrm{ex}^2 \over L^2} |\nabla \mm|^2 +
Q \int_\Omega |\mathbf m_\perp|^2 - \int_\Omega \left( 2 \mathbf
  h_\mathrm{ext} \cdot \mm - \mathbf h \cdot \mm + h_\mathrm{ext}^2
\right),
\end{align}
where $\Omega = (0,1) \times \TTT$ is the rescaling of
$\widetilde\Omega$, with $\TTT = [0, \ell)^2$, and $|\mm| =
\chi_\Omega$, where $\chi_\Omega$ denotes the characteristic function
of $\Omega$. The dimensionless stray field $\hh$ is defined as the
unique (see e.g. \cite{choksi98}) solution in $L^2(\mathbb R \times
\TTT; \mathbb R^3)$ of
\begin{align} \label{def-h}
\nabla \times \mathbf h = 0, \qquad \nabla \cdot \mathbf h = -
\nabla \cdot \mm && \mathrm{in} \quad \mathbb R \times \TTT,
\end{align} 
where \eqref{def-h} is understood in the distributional sense.

\medskip

{\it Sharp interface energy:} In a bulk uniaxial material, the
magnetization is expected to lie mostly in the direction of the easy
axis, i.e. $\mm \approx \pm \mathbf e_1$. The regions with different
orientations of the magnetization are expected to be separated by thin
Bloch walls \cite[p. 215]{hubert}. A Bloch wall is characterized by a
transition layer of thickness $w \approx Q^{-\frac 12} l_\mathrm{ex}$ in
which the magnetization rotates in the wall plane, thus avoiding the
creation of a stray field.  Taking advantage of the observation in
\cite[p. 367]{kohn07iciam}, one
% I want to emphasize this reference
can directly estimate the anisotropy and exchange terms in the energy
from below for any $\delta \ge 0$ fixed as
\begin{align} \label{eq:9} %
\int_{|\mm_\perp| \geq \delta} \left( {l_\mathrm{ex}^2 \over
    L^2} |\nabla \mm|^2 + Q |\mm_\perp|^2 \right) \geq \
\eps \int_{|\mm_\perp| \geq \delta} |\nabla m_1| ,
\end{align}
with the notation $\mm = m_1 \mathbf e_1 + \mm_\perp$ and
where we have introduced
\begin{align}
\label{eq:10}
\eps = {2 l_\mathrm{ex} \sqrt{Q} \over L}. 
\end{align}
In fact, the one-dimensional Bloch wall profile attains equality in
\eqref{eq:9} (see \cite{hubert}), which implies that the term in the
right-hand side of \eqref{eq:9} should actually well approximate the
term in the left-hand side for the energy minimizers. The condition of
validity of this approximation is that the wall thickness $w$ remains
much smaller than the characteristic length scale of the magnetization
pattern. In particular, one should have $w \ll L$. This condition is
achieved for sufficiently thick plates. In fact, large thickness is
also a necessary condition for branched domain patterns to be observed
in ferromagnetic materials \cite{hubert}. Therefore, in the present
context one is naturally interested in the asymptotic behavior of
energy for large values of $L$ or, equivalently, in the limit $\eps
\to 0$ with all other dimensionless parameters fixed.

\medskip

Dropping the gradient term in (\ref{eq:8}) where $|\mm_\perp| <
\delta$ and combining it with (\ref{eq:9}), one can see that
\begin{align}
 \EEE [\mm] \geq \ \eps \int_{|\mm_\perp| \geq \delta} |\nabla
  m_1| \, + Q \int_{|\mm_\perp| < \delta} |\mm_\perp|^2 - \int_\Omega
 \left( 2 \mathbf h_\mathrm{ext} \cdot \mm + \mathbf h \cdot \mm -
  h_\mathrm{ext}^2 \right) .  \label{eq:22}
\end{align}
This motivates the introduction of a sharp interface energy, in which
the gradient-squared term in \eqref{eq:8} is replaced by the total
variation of $m_1$ (see also \cite{choksi98,choksi99}). We note,
however, that the sharp interface energy is basically a tool to
approximate the behavior of the full physical energy in \eqref{eq:8}
and, therefore, can be tailored to our advantage. We choose the sharp
interface energy in the form
\begin{align} \label{eq:5} %
  E[\mm] = \eps \int_\Omega |\nabla m_1^\delta| + Q \int_{\{|\mathbf
    m_\perp| < \delta\}}|\mm_\perp|^2 + \delta^2 Q \int_{\{|\mm_\perp|
    \geq \delta\}}
 |\mm_\perp|^2 
 - \int_\Omega \left( 2 \mathbf h_\mathrm{ext} \cdot \mm + \mathbf h
  \cdot \mm - h_\mathrm{ext}^2 \right).
\end{align}
Here $0 < \delta \ll 1$ is an arbitrary ``cutoff'' parameter, whose
precise value is inessential (hence the index $\delta$ is dropped from
the definition of $E$), and $m_1^\delta$ is the truncated version of
$m_1$:  
\begin{align}
  \label{eq:2}
  m_1^\delta(x) = \left\{
  \begin{array}{cc}
    1 - \delta^2, & m_1(x) > 1 - \delta^2, \\
    m_1(x), & -1 + \delta^2 \leq m_1(x) \leq 1 - \delta^2, \\
    -1 + \delta^2, & m_1(x) < -1 + \delta^2.
  \end{array} \right.
\end{align}
The advantage of using $m_1^\delta$ in (\ref{eq:5}) instead of $m_1$
is that, consistently with (\ref{eq:22}), the interfacial term does
not contribute to the energy away from the domain walls, where $m_1
\approx \pm 1$.  Importantly, $E$ provides an ansatz-free lower bound
for $\EEE$:
\begin{align}
  \label{eq:18}
  \EEE[\mm] \geq (1 - \delta^2) E[\mm],
\end{align}
which can be easily seen by retracing the arguments leading to
\eqref{eq:22}.  Furthermore, since in the limit $\eps \to 0$ the
transition regions between different directions of $\mm$ are expected
to become $O(\eps Q^{-1})$ thin and the inequality in (\ref{eq:22}) to
become an equality for minimizers, in view of arbitrariness of
$\delta$ one should expect that $\inf \EEE \simeq \inf E$ for $\eps
\ll 1$. In the following, we will prove that this relation holds in
the sense of scaling, i.e., for sufficiently small $\eps > 0$, we have
$\inf \EEE \sim \inf E$.

\medskip

{\it Critical external fields:} Clearly, when the applied field
$h_\mathrm{ext}$ is sufficiently large, the minimal energy
configuration will be such that all magnetic moments are aligned with
the field, i.e. $\mm = \mathbf e_1 \chi_\Omega$. For smaller external
fields, the minimizer is attained by other configurations. The
external field strength, at which the uniform magnetization $\mm =
\mathbf e_1 \chi_\Omega$ looses its optimality is denoted by
$h_{c_1}$. Let us also note that appearance and disappearance of
patterns as a function of the control parameter in systems of this
kind is often accompanied by a hysteresis. Therefore, non-trivial
critical points of the energy may persist even for fields larger than
$h_{c_1}$. The critical field at which these critical points disappear
will be denoted by $h_{c_0}$.

\medskip

The saturation field $h_s$ is defined similarly in terms of the {\it
  relaxed energy}. In our setting, it is the variant of $\EEE$ where
the surface energy is not penalized, i.e.
\begin{align*} %
  \EEE_\mathrm{rel}[\mm] = \int_\Omega Q |\mm_\perp|^2 - \int_\Omega
  \left( 2 \mathbf h_\mathrm{ext} \cdot \mm - \mathbf h \cdot \mathbf
    m + h_\mathrm{ext}^2 \right).
\end{align*}
The set of admissible functions for $\EEE_\mathrm{rel}$ is given by
all $\mm$ satisfying $|\mm| \leq 1$ in $\Omega$. This relaxed
constraint in the above calculation can be justified by looking at
small-scale oscillations of $\mm$.  It is related to the fact that
$\EEE_\mathrm{rel}$ is non-convex, see e.g.  \cite{Dacorogna-Book,
  KohnStrang-1986}.  One expects that the relaxed energy gives the
leading order behavior of the minimal energy, i.e.
\begin{align*}
 \lim_{\eps \to 0} \inf_{|\mm| = 1} \EEE[\mm] \ = \ \inf_{|\mm| \leq 1}
 \EEE_\mathrm{rel}[\mm].
\end{align*}
The saturation field strength $h_s$ is defined as the field strength
at which $\mm = \mathbf e_1 \chi_\Omega$ looses its energetic
optimality in terms of the relaxed energy $\EEE_\mathrm{rel}$ and is
expected to be close to $h_{c_1}$ when $\eps \ll 1$.

\medskip

To understand better the behavior of $\EEE_\mathrm{rel}$, let us first
introduce the notation for the average of a quantity $f = f(x_1, x_2,
x_3)$ over $\TTT$ at fixed $x_1 \in \mathbb R$.  We use the notation
$\overline f(x_1) \ := \ {1 \over \ell^2} \int_\TTT f(x_1, \cdot)$. We
note that the solutions of (\ref{def-h}) have the following basic
properties (the proof is by an elementary integration by parts):
\begin{lemma} \label{l:mh} %
 Let $\mathbf h \in L^2(\mathbb R \times \TTT; \mathbb R^3)$ be a
 solution of (\ref{def-h}). Then
 \begin{gather} \label{eq:53} %
   \overline h_1(x_1) = - \overline m_1(x_1) \qquad \text{ for a.e. }
   x_1 \in \mathbb R,\\
   \label{eq:57}
   \int_\Omega \mathbf h \cdot \mm = - \int_{\mathbb R \times \TTT}
   |\mathbf h|^2, \qquad\qquad \int_{\R \times \TTT} |\hh|^2 \ \leq
   \int_{\Omega} |\mm|^2.
 \end{gather}
\end{lemma}
% \begin{proof}
%   Integrating (\ref{def-h}) over $(x_2, x_3) \in \TTT$ with $x_1 \in
%   \mathbb R$ fixed (approximating $\mm$ by smooth functions in
%   $L^2(\mathbb R \times \TTT; \mathbb R^3)$, if necessary), by
%   periodicity of $\mm$ and $\mathbf h$ we obtain that
%   \begin{align*}
%     \partial_1 \overline m_1 + \partial_1 \overline h_1 = 0.
% \end{align*}
% The result in (\ref{eq:53}) then follows by integrating the obtained
% expression over $(-\infty, x_1)$ and taking into account that both
% $\mm$ and $\mathbf h$ vanish at $x_1 = -\infty$.

% \medskip

% To prove (\ref{eq:57}), we write $\mathbf h = -\nabla \varphi$ for
% some $\varphi: \mathbb R \times \TTT \to \mathbb R$, such that
% \begin{align}
%   \label{eq:phi}
%   \Delta \varphi = \nabla \cdot \mm, \qquad \mathrm{in} \qquad
%   {\mathbb R \times \TTT}. 
% \end{align}
% Once again, approximating $\mm$ by smooth functions, if
% necessary, and integrating by parts, we obtain
% \begin{align}
%   \int_\Omega \mathbf h \cdot \mm = - \int_{\mathbb R \times \TTT}
%   \nabla \varphi \cdot \mm = \int_{\mathbb R \times
%     \TTT} \varphi \, \nabla \cdot \mm \notag \\
%   = \int_{\mathbb R \times \TTT} \varphi \Delta \varphi = -
%   \int_{\mathbb R \times \TTT} |\nabla \varphi|^2 = - \int_{\mathbb R
%     \times \TTT} |\mathbf h|^2,
% \end{align}
% where we took into account (\ref{eq:phi}) and that $\mathrm{supp} \,
% \mm = [0, 1] \times \TTT$.
% \end{proof}

Using \eqref{eq:53} and \eqref{eq:57}, one easily computes that
\begin{align} \label{eq:52} %
  \EEE_\mathrm{rel}[\mm] \ %
  &\lupupref{eq:53}{eq:57}\geq \ \int_\Omega (h_\mathrm{ext} + h_1)^2
  \ %
  \lupref{eq:53}\geq \ \int_\Omega (h_\mathrm{ext} - \overline m_1)^2,
\end{align}
where we have used Jensen's inequality in the second
inequality. Hence,
\begin{align} \label{eq:20} %
  \EEE_\mathrm{rel}[\mm] \geq \inf_{|\mm| \leq 1} \int_\Omega (
  h_\mathrm{ext} - \overline m_1)^2 =
 \begin{cases}
   0, & 0 < h_\mathrm{ext} \leq 1, \\
   \ell^2 (1 -h_\mathrm{ext})^2, & h_\mathrm{ext} > 1.
 \end{cases}
\end{align}
On the other hand, equality in \eqref{eq:20} is achieved by using the
trial function $\mm = \min\{1, h_\mathrm{ext} \} \, \ee_1
\chi_\Omega$.  Therefore, $\mm = \ee_1 \chi_\Omega$ is the minimizer
of \eqref{eq:20} if and only if $h_\mathrm{ext} \geq h_s = 1$, which
is precisely the saturation field.

\medskip

Since we are interested in the bifurcation from the uniform to a
patterned magnetization occurring near saturation, we introduce a
parameter $\lambda$ which measures the deviation from saturation:
\begin{align} \label{eq:24} %
 h_\mathrm{ext} = 1 - \lambda,
\end{align}
where $0 < \lambda \ll 1$ means the applied field is just below the
saturation threshold. One question we want to address is how to
calculate $\lambda_{c_0}$ and $\lambda_{c_1}$ corresponding to the
critical fields $h_{c_0}$ and $h_{c_1}$.

\medskip

{\it Reformulation of the sharp interface energy:} We now derive an
expression for energies $\EEE$ and $E$ in new variables which make our
analysis more convenient. We introduce
\begin{align}
\label{def-uv}
\uu = \mm - (1 - \lambda) \chi_\Omega \mathbf e_1,
&& \mathbf v = \mathbf h + (1 - \lambda) \chi_\Omega \mathbf
e_1. 
\end{align}
Then, using Lemma \ref{l:mh}, \eqref{eq:24}, \eqref{eq:57} and
\eqref{def-uv}, one gets
\begin{align*}
  \int_\Omega \left( h_\mathrm{ext}^2 - \hh \cdot \mm - 2
    {\hh_\mathrm{ext}} \cdot \mm \right)
  &= \ \int_{\R \times \TTT} |\vv|^2 - 2 (1-\lambda) \int_{\Omega}
  (v_1 + u_1) \ \upref{eq:53}= \ \int_{\R \times \TTT} |\vv|^2.
\end{align*}
Therefore, we can rewrite the energy $E$ from \eqref{eq:5} as follows
(with a slight abuse of notation, we view $E$ from now on as a
function of $\uu$ instead of $\mm$)
\begin{align} \label{energy} %
E[\uu] = \eps \int_\Omega |\nabla u_1^\delta | + Q
\int_{\{|\uu_\perp| < \delta\}} |\uu_\perp|^2 + \int_{\mathbb R
  \times \TTT} |\mathbf v|^2 + \delta^2 Q \int_{\{|\uu_\perp| \geq
 \delta\}} |\uu_\perp|^2,
\end{align}
where, as before, $u_1$ and $\uu_\perp$ denote the components of
$\uu$ along and normal to the easy axis, respectively,
$u_1^\delta = m_1^\delta - 1 + \lambda$, and $\mathbf v$ solves
\begin{align} \label{def-v} %
\mathbf v = -\nabla \varphi, \qquad \Delta \varphi = \nabla \cdot
\uu && \mathrm{in} \quad \mathbb R \times \TTT,
\end{align}
 The set of admissible functions for \eqref{energy} is given by
\begin{align*} %
  \AA = \big \{ \uu \in BV({\R \times \TTT}; \R^3) \ : \ |\uu + (1 -
  \lambda) \chi_\Omega \ee_1 |= \chi_\Omega \big \}.
\end{align*}
Similarly, the expression in \eqref{eq:8} can be rewritten as
\begin{align}
  \label{eq:45}
 \EEE[\uu] = \int_\Omega {\eps^2 \over 4 Q} |\nabla \uu|^2 + Q
    \int_\Omega |\uu_\perp|^2 + \int_{\mathbb R \times \TTT} |\mathbf
  v|^2. 
\end{align}
For simplicity of notation, we take the same admissible class $\AA$
for $\EEE$ as well, setting $\EEE[\uu] = +\infty$, whenever
$\uu_{|\Omega} \not\in H^1(\Omega)$.

\section{Scaling of the energy in bulk samples}
\label{sec-bulk}

In this section, we investigate the scaling behavior of the energy of
minimizers in the case of bulk samples corresponding to the limit
$\eps \to 0$. The main part of this section will be concerned with the
sharp interface energy $E$ defined in \eqref{energy}. The connection
to the diffuse interface energy $\mathcal E$ is then shown in Section
\ref{ss-diffuse}.

\medskip

The model has three dimensionless parameters: $\eps$, $\lambda$,
$Q$. In particular, we are interested in the case of macroscopically
large samples near critical fields, i.e. $\eps \ll 1$ and $\lambda \ll
1$. Our result shows that for sufficiently small $\eps$ and $\lambda$
with fixed $Q$ there are exactly two different scaling regimes, each
corresponding to a particular pattern of magnetization attaining the
minimal energy scale. Introducing
\begin{align} \label{def-gamma} %
\gamma = \frac Q{1 + Q},
\end{align}
we have the following result for the sharp interface energy $E$:
\begin{theorem} \label{thm-bulk} %
  Let $\lambda \lesssim \gamma^2 |\ln \lambda|^2$ and $\ell \gtrsim
  \gamma^{-\frac 13} \eps^{\frac 13} \lambda^{-\frac 12} |\ln
  \eps|^{-\frac 13}$. Then for $\eps \ll 1$ and $\lambda \ll 1$, we
  have
 \begin{equation*} %
   \frac 1{\ell^2} \inf_{\uu \in \AA} E[\uu] \ \sim \ %
   \min \left \{ \lambda^2, \ \gamma^{\frac 13} \eps^{\frac 23} \lambda |\ln
    \lambda|^{\frac 13} \right\}.
 \end{equation*}
\end{theorem}
The first regime corresponds to a uniform magnetization along the
applied field, while the second regime is achieved by branched
magnetization patterns. Note that as long as $Q \gtrsim 1$, the
particular value of $Q$ does not affect the scaling of the minimal
energy. This indicates that the restricted model corresponding to $Q =
\infty$, i.e.  when $\mm = \pm \mathbf e_1 \chi_\Omega$, captures the
essential features of the general model in \eqref{energy}. On the
other hand, for $Q \ll 1$ the effect of anisotropy only has the effect
of renormalizing the minimal energy scaling by a factor of $Q^{\frac
  13}$. This will be further discussed in Sec. \ref{sec-reduced} with
the help of a reduced sharp interface model.

\medskip

Combining the results in Theorem \ref{thm-bulk} with \eqref{eq:18} and
the constructions of Sec. \ref{ss-diffuse}, the full micromagnetic
energy $\mathcal E$ satisfies the same scaling
\begin{theorem} \label{thm-diffuse} %
  Let $\lambda \lesssim \gamma^2 |\ln \lambda|^2$ and $\ell \gtrsim
  \gamma^{-\frac 13} \eps^{\frac 13} \lambda^{-\frac 12} |\ln \eps|^{-\frac 13}$. Then for
  $\eps \ll 1$ and $\lambda \ll 1$, we have
  \begin{align*} %
    \frac 1{\ell^2} \inf_{\uu \in \mathcal A} \EEE[\uu] \ \sim \
    \begin{cases}
      \lambda^2 & \text{for } \ \lambda \ \lesssim \ \gamma^{\frac 13}
      \eps^{\frac 23} |\ln
      \eps|^{\frac 13}, \\
      \gamma^{\frac 13} \eps^{\frac 23} \lambda |\ln \lambda|^{\frac
        13} \qquad&\text{for } \ \lambda \ \gtrsim \ \gamma^{\frac 13}
      \eps^{\frac 23} |\ln \eps|^{\frac 13}.
    \end{cases}
 \end{align*}
\end{theorem}
This theorem implies that for small enough values of $\lambda(\eps)$
the minimal energy scaling is achieved by uniform magnetization
pattern (the monodomain state: $\mm = \ee_1 \chi_\Omega$), while for
sufficiently large values of $\lambda(\eps)$ the optimal energy
scaling is achieved by a branched domain pattern, as $\eps \to 0$. The
transition occurs at $\lambda_c \sim \gamma^{\frac 13} \eps^{\frac 23}
|\ln \eps|^{\frac 13}$.

\medskip 

The analysis techniques we employ in this section go back to the work
of Choksi and Kohn in \cite{choksi98,choksi99}, who analyzed the
energy of ferromagnetic plates in the absence of a magnetic field. In
our analysis we identify the optimal dependence of the minimal energy
on the parameter $\lambda$ which is not addressed in \cite{choksi98,
  choksi99}. In our analysis, we also apply tools from related works
in the framework of type-I superconductors \cite{choksi08}.  There the
authors derive the scaling of the energy for the type-I superconductor
near critical field. We note that the super conductor model is more
rigid, since there the two different phases are described by the
characteristic function $\chi$ which only takes the discrete values 0
and 1 and a divergence free magnetic field $B$, whereas in our model
the magnetization $\mm$ is allowed to take all values on the unit
sphere.

\subsection{Preliminaries}  %

In this section we collect some useful results before addressing the
proof of the upper and lower bound in the next two sections.

\medskip

Control on \eqref{energy} yields information about $u_1$ and $v_1$ on
each slice. As expected, the stray field favors zero average of $u_1$
on each tangential slice:
\begin{lemma} \label{lem-slice} %
  Let $\uu \in \AA$ and let $\vv \in L^2(\R \times \TTT; \R^3)$
  satisfy \eqref{def-v}. Then for every $c_1 > 0$ and $c_2 > 0$ there
  exists a constant $c > 0$, such that if $E[\uu]/\ell^2 \leq c
  \lambda^2$ then we have
  \begin{equation} \label{eq-l2} %
    \int_0^1 \ |\overline u_1|^2 dx_1 \ \leq \ c_1^2 \lambda^2.
\end{equation}
Furthermore there exists $I \subseteq (0,1)$ with $|I| > 1 - c_2$,
such that for all $a \in I$
\begin{equation} \label{eq-l3} %
|\overline u_1(a)| \ \leq \ c_1 \lambda.
\end{equation}
\end{lemma}

\begin{proof}
  We first note that in terms of $\uu$ and $\vv$, in view of
  \eqref{eq:53}, we have $\bar v_1 = - \bar u_1$ for a.e. $a \in \R$.
  By Jensen's inequality, it then follows that
  \begin{equation*}
    \ell^2 \int_0^1 |\overline u_1|^2 dx_1  \ = \ \int_{\Omega} 
    |\overline u_1|^2  \ = \ \int_{\Omega} |\overline \vv|^2  
    \ \leq \  \int_{\Omega} |\vv|^2 \ \leq \ E[\uu],
  \end{equation*}
  and \eqref{eq-l2} follows.  Inequality \eqref{eq-l3} follows from
  \eqref{eq-l2} by an application of Fubini's Theorem.
\end{proof}

The main ingredient for the proof of the lower bound is an estimate
that characterizes the transition energy, i.e. the cost for the
magnetization to vary between a tangential slice $\{a \} \times \TTT$
and its value zero outside of the sample. The idea to estimate such
transition energies was introduced in \cite{KohnMueller-1992} and has
been subsequently applied also in e.g. \cite{choksi04,choksi08}.

\begin{lemma}[Transition energy] \label{lem-transition} %
  For every $\uu \in \AA$ there exists $I \subset (0,1)$ with $|I| >
  \frac 12$, such that for all $a \in I$ and for all $\psi \in H^1(\TTT)$,
  we have
\begin{align*}
  \left| \int_{\{a \} \times \TTT} u_1 \psi \right| \ \lesssim \
  E^{\frac 12} [\uu] \left( \gamma^{-\frac 12} \nltL{\nabla
      \psi}{{\TTT}} + \nltL{\psi}{{\TTT}} \right).
\end{align*}
\end{lemma}
\begin{proof}
  Let us first assume that $\uu \in \cciL{\R \times \TTT}$, with $\uu
  = 0$ outside of $[-1,2] \times \TTT$. Let $\vv$ be defined by
  \eqref{def-h}. Noting that $\psi$ does not depend on $x_1$ and using
  integration by parts, for any $a \in (0,1)$ and $b \in (-2,-1)$, we
  then get
  \begin{align}
    \hspace{2ex}& \hspace{-2ex} %
    \int_{\TTT} \ u_1(a, \cdot) \psi \ = \ \int_b^a \int_\TTT
    \partial_1 \big( u_1(x_1, \cdot) \, \psi \big) \, dx_1 \notag \\
    &\upref{def-h}= \ - \int_b^a \int_\TTT \nabla_\perp \cdot
    \uu_\perp(x_1, \cdot) \ \psi \, dx_1 - \int_b^a \int_\TTT
    \nabla_\perp \cdot \vv_\perp(x_1, \cdot) \ \psi \, dx_1 - \int_b^a
    \int_\TTT \partial_1 v_1(x_1, \cdot) \
    \psi \, dx_1 \notag  \\
    &= \ \int_b^a \int_\TTT \uu_\perp(x_1, \cdot) \cdot \nabla_\perp
    \psi \, dx_1 + \int_b^a \int_\TTT \vv_\perp(x_1, \cdot) \cdot
    \nabla_\perp \psi \, dx_1 + \int_{\TTT} v_1(a, \cdot) \psi \, dx_1
    - \int_{\TTT} v_1(b, \cdot) \psi. \label{trans-1}
  \end{align}
  By Fubini's theorem, there exists $b \in (-2,-1)$ and $I \subseteq
  (0,1)$ with $|I| > \frac 12$ such that for all $a \in I$,
  \begin{align} \label{trans-2} %
    \int_{\{a \} \times \TTT} |v_1|^2 + \int_{\{b \} \times \TTT}
    |v_1|^2 \ \lesssim \ \int_{\R \times \TTT} |v_1|^2 .
  \end{align}
  The statement then follows for all $a \in I$ from \eqref{trans-1},
  \eqref{trans-2} and by application of Cauchy-Schwarz inequality and
  \eqref{energy}.

\medskip

Now consider a general $\uu \in \AA$. In this case, $\uu$ can be
approximated by a sequence of functions $\uu^j \in \cciL{\R \times
  \TTT}$ such that $\uu^j = 0$ outside of $[-1,2] \times \TTT$ and
such that $\uu^j \to \uu$ in $L^2(\R \times \TTT; \R^3)$ and $\int_{\R
  \times \TTT} |\nabla \uu^j| \to \int_{\R \times \TTT} |\nabla \uu|$,
see \cite{EvansGariepy-Book}.  By \eqref{eq:57}, we also have $\vv^j
\to \vv$ in $L^2(\R \times \TTT; \R^3)$, where $\vv^j$ denotes the
stray field of $\uu^j$. Taking a subsequence, if necessary, we also
have convergence $\uu^j \to \uu$, $\vv^j \to \vv$ in $L^2(\{a \}
\times \TTT; \R^3)$ for a. e. $a \in \R$. Using this approximation,
the lemma follows.
\end{proof}

We will also use the following technical lemma of De Giorgi (see,
e.g., \cite[Lemma 3.1]{choksi08}):
\begin{lemma} \label{lem-iso} Let $S \subset \TTT$ be a set of finite
  perimeter, and let $r > 0$ be such that $r |\partial S| \leq \frac
  14 |S|$.  Then there exists an open set $\overline S \subset \TTT$
  with the properties
  \begin{enumerate}
  \item[(i)] There is a considerable overlap of $\bar S$ with $S$, in
    the sense of $|S \cap \overline S| \geq \frac 12 |S|$.
  \item[(ii)] For all $t > 0$, the set $\overline S^t \ := \ \big \{ p
    \in \TTT : \dist(p, \overline S) < t \big \}$ satisfies
    $|\overline S^t| \ \lesssim \ |S| (1+ (t/r)^2)$.
  \end{enumerate}
\end{lemma}

\subsection{Ansatz-free lower  bound} 

In this section, we present the proof for the lower bound. We need to
show:
\begin{proposition}[Lower bound] \label{prp-1} %
 For $\eps \ll 1$ and $\lambda \ll \gamma$, we have
 \begin{equation*} %
   {1 \over \ell^2} \inf_{\uu \in \AA} E[\uu] \ \gtrsim \ \min \left
     \{ \lambda^2, \ \gamma^{\frac 13} \eps^{\frac 23} \lambda |\ln
     \lambda|^{\frac 13} \right\}.
  \end{equation*}
\end{proposition}

\begin{proof}
  Following the ideas in \cite{choksi08}, we argue as follows. Recall
  that in view of Lemma \ref{lem-transition}, the energy is bounded
  below by a Sobolev-type norm of negative order on $\uu$, evaluated
  on a generic tangential slice $\{ a \} \times T$. In this proof, we
  combine this with control that we have on the surface energy and
  anisotropy energy on a generic slice. The proof is divided into five
  steps.

\medskip

{\it Step 1: Identification of tangential slice.} We will argue by
contradiction.  Hence, we may assume that the energy does {\it not}
satisfy the lower bound, i.e., there exists $\uu \in \AA$, such that
\begin{align} \label{E-small} %
 E[\uu] \ \ll \ \ell^2 \min \left \{ \lambda^2, \
    \gamma^{\frac 13} \eps^{\frac 23} \lambda |\ln \lambda|^{\frac 13} \right\},
  \end{align}
  for some $\eps \ll 1$ and $\lambda \ll \gamma$.  We choose $a \in
  (0,1)$ such that the assertions of Lemma \ref{lem-slice} and Lemma
  \ref{lem-transition} hold. By Lemma \ref{lem-slice} and by Fubini's
  Theorem, we may then assume that $a$ is furthermore chosen such that
  \begin{gather} %
    |\overline u_1(a)| \ \ll \ \lambda \label{av-small} \\
   \int_{\{a \} \times \TTT} |\nabla u_1| \ \lesssim \
    \frac{E}\eps \ \ \upref{E-small}\ll \ \ell^2 \min \left \{
      \frac{\lambda^2}\eps, \ \frac{\gamma^{\frac 13} \lambda |\ln
        \lambda|^{\frac 13}}{\eps^{\frac 13}}
    \right\},  \label{surface-small} \\
   \int_{\{a \} \times \TTT} |\uu_\perp|^2 \ \lesssim \
    \frac{E}{Q} \ \ \upref{E-small}\ll \ \ell^2 \min \left \{ \frac
      {\lambda^2}Q, \ \frac {\gamma^{\frac 13} \eps^{\frac 23} \lambda |\ln
        \lambda|^{\frac 13}}Q \right\}. \label{anisotropy-small}
  \end{gather} {\it Step 2: Structure of magnetization.} We next
  analyze the magnetization on the slice $A := \{a \} \times \TTT$ in
  more detail. In view of the upper constructions, we expect that the
  regions where the magnetization points in the negative
  $x_1$-direction are small needle-shaped domains. The restriction of
  $m_1$ to the slice, therefore, is expected to be negative on a
  number of small circular domains. In the following, we give a
  precise version of this heuristic picture. We define the set $A_+$
  (where $\mm$ points ``to the right'', i.e. in direction of $\ee_1$.
  The notation ``to the right'' is in accordance with the figures) and
  the set $A_-$ (where $\mm$ points to the left) by
  \begin{gather*}
    A_+ \ := \ \big \{ x \in A \ : \ \frac 12 \lambda < u_1 \leq
    \lambda \big \}, \qquad\qquad A_- \ := \ \big \{ x \in A \ : \ -2
    + \lambda \leq u_1 < -1 + \lambda \big \}.
  \end{gather*}
  It is also convenient to define a transition region $A_0$ (where
  $\mm$ also points to the right)
  \begin{align*}
    A_0 := \big \{ x \in A \ : \ -1 + \lambda \leq u_1 \leq \frac 12
    \lambda \big \}.
  \end{align*}
  We first note that when $m_1 \geq 0$, we have $|\uu_\perp|^2 =
  |\mm_\perp|^2 = 1 - m_1^2 \geq 1 - m_1 = \lambda - u_1$, i.e.
  \begin{align}
    \label{eq:8-2}
    |\uu_\perp|^2 \ \geq \ \lambda - u_1 && \text{in } A_0 \cup A_+.
  \end{align}
  We claim that the region $A_-$ of ``reversed magnetization'' is
  concentrated on a small set with total area of order $\lambda
  \ell^2$, and that the transition region is even smaller.  More
  precisely, we claim that
  \begin{align} \label{geometry} %
    |A_+| \ \sim \ \ell^2, && |A_0| \ \ll \ \lambda \ell^2, && |A_-| \
    \sim \ \lambda \ell^2.
  \end{align}
  Indeed, by (\ref{eq:8-2}), we have $|\uu_\perp|^2 \geq
  \frac{\lambda}{2}$ in $A_0$, and so
  \begin{align*}
    |A_0| \ \leq \ {2 \over \lambda} \int_{A_0} |\uu_\perp|^2 \
    \lesssim \ {E \over \lambda Q} \ \lupref{anisotropy-small}\ll \
    \ell^2.
  \end{align*}
  Here we also have used that by assumption $Q \geq \gamma \gg
  \lambda$. Choosing $c_1 = {1 \over 16}$ in Lemma \ref{lem-slice},
  and in view of $|A_0| + |A_+| = \ell^2 - |A_-|$, we get
  \begin{align*}
    -\tfrac{1}{16} \lambda \ell^2 \ \lupref{eq-l3}\leq \ \int_A u_1 \
    \leq \ (- 1 + \lambda) |A_-| + \lambda (|A_0| + |A_+|) \ = \ -
    |A_-| + \lambda \ell^2,
  \end{align*}
  so $|A_-| \leq {17 \over 16} \lambda \ell^2 \ll \ell^2$ and,
  therefore, $|A_+| \geq \tfrac12 \ell^2$. Similarly,
  \begin{align*}
    -\tfrac{1}{16} \lambda \ell^2 \ \ &\upref{eq-l3}\leq \ - \int_A
    u_1 \ \leq \ 2 |A_-| +
    \int_{A_0} (\lambda - u_1)  - \tfrac12 \lambda |A_+| \notag \\
    &\upref{eq:8-2}\leq \ 2 |A_-| + \int_{A_0} |\uu_\perp|^2 -
    \tfrac14 \lambda \ell^2 \leq \ 2 |A_-| + {C E \over Q} - \tfrac14
    \lambda \ell^2 \lupref{anisotropy-small}\leq 2 |A_-| - \tfrac18
    \lambda \ell^2.
  \end{align*}
  Hence $|A_-| \geq {1 \over 32} \lambda \ell^2$, and so $|A_-| \sim
  \lambda \ell^2$.  This concludes the proof of \eqref{geometry}.

  \medskip

  {\it Step 3: Identification of a regularized region.} In the
  previous step, we showed that the reversed magnetization region
  $A_-$ occupies a small fraction of $A$.  By the upper bound
  constructions of Lemma \ref{lem-needle} below, we would expect that
  $A_-$ is divided into a controlled number of similar size circular
  domains. The construction also suggests that in the core region of
  the plate, the typical radius $r$ of these circular domains and the
  typical distance $a$ between them are given by
  \begin{align} \label{def-rl} %
    r \ = \ \frac {\eps^{\frac 13}}{\gamma^{\frac 13} |\ln \lambda|^{\frac 13}}, && a
    \ = \ \frac {\eps^{\frac 13}}{\gamma^{\frac 13} \lambda^{\frac 12} |\ln
      \lambda|^{\frac 13}}.
  \end{align}
  In the following, we use the co-area formula and the isoperimetric
  inequality to get a rigorous variant of the above heuristics. We
  replace $A_-$ by a larger set $S$ with $A_- \subseteq S \subseteq
  A_- \cup A_0$.  The reason to choose $S$ instead of $A_-$ is that we
  cannot exclude a concentration of surface energy on $\partial
  A_-$. We claim that there is $c \in (-\frac12, -\frac14)$ such that
  the set
  \begin{align} \label{def-S} %
    S \ := \ \big \{ \ x \in A \ : \ u_1 < c \ \big \}
  \end{align}
  satisfies
  \begin{gather}
    |S| \ \sim \ \lambda \ell^2, \label{S-set} \\
    |\partial S| \ \ \upref{surface-small} \ll \ \ell^2 \min \left \{
      \frac{\lambda^2}\eps, \ \frac{\gamma^{\frac 13}\lambda |\ln
        \lambda|^{\frac 13}}{\eps^{\frac 13}} \right\}.
    \label{S-bdry}
  \end{gather}
  Indeed, by \eqref{def-S} it follows that $A_- \subseteq S \subseteq
  A_- \cup A_0$ and \eqref{S-set} follows by
  \eqref{geometry}. Furthermore, by the co-area formula
  \begin{align*}
    \int_A |\nabla_\perp u_1| \ = \ \int_\R \mathcal H^1(\{ x_\perp
    \in A : u_1(x_\perp) = t \}) \ dt,
  \end{align*}
  and Fubini's Theorem, there exists $c \in (-\frac12, -\frac14)$,
  such that $S$ satisfies
  \begin{align*}
    |\partial S| \ \lesssim \ \int_A |\nabla_\perp u_1| \
    \upref{surface-small}\ll \ \ell^2 \min \left \{
      \frac{\lambda^2}\eps, \ \frac{\gamma^{\frac 13}\lambda |\ln
        \lambda|^{\frac 13}}{\eps^{\frac 13}} \right\}.
  \end{align*}
  Estimates \eqref{S-set} and \eqref{S-bdry} together with
  \eqref{def-rl} yield
  \begin{align} \label{iso-r} r \ |\partial S| \ \ll \ |S|.
  \end{align}
  % Estimate \eqref{iso-r} yields control how much $S$ can be
  % different from a single ball.
  Heuristically, the estimate in (\ref{iso-r}) means that $S$, roughly
  speaking, splits into a collection of disks of diameter much larger
  than $r$. This disagrees with the expectation from the upper
  construction and will lead to a contradiction.

  \medskip

  We next replace $S$ by another set $\overline S$, still satisfying
  all the relevant properties of $S$. Additionally, it grows in a
  controlled way upon ``thickening''. More precisely, in view of
  \eqref{iso-r} and by Lemma \ref{lem-iso}, there is a set $\overline
  S$ with
  \begin{align} \label{ss-1} %
   |S \cap \overline S| \ \geq \ \frac 12 |S|
    \ \upref{S-set}\gtrsim \ \lambda \ell^2
  \end{align}
  and such that for all $t \geq r$, the $t$-neighborhood $\overline
  S^t$ of $\overline S$ satisfies
  \begin{align} \label{ss-2} %
    |\overline S^t| \ \lesssim \ \frac {t^2}{r^2} |S| \ %
    \lupref{S-set}\lesssim \ \frac {\lambda t^2 \ell^2}{r^2}.
  \end{align}
  {\it Step 4: Definition of a suitable test function.} We now define
  a logarithmic cut-off $\psi \in H^1(\TTT)$, with $0 \leq \psi \leq
  1$, around $\overline S$. Let
  \begin{align} \label{def-psi} %
    \psi(x_\perp) \ := \ \phi(\dist(x_\perp,\overline S)), \qquad
    \text{ where } \qquad \phi(t) \ := \ \left \{
      \begin{array}{ll}
        1 & \text{for } 0 \leq t \leq r, \\
        \frac {\ln (a/t)}{\ln(a/r)} & \text{for } r \leq t \leq a, \\
        0 & \text{for } a \leq t,
      \end{array}
    \right.
  \end{align}
  with $r$ and $a$ defined in \eqref{def-rl}.  A direct computation,
  following \cite{choksi08}, then yields
  \begin{align} %
    \int_\TTT \psi \ \lesssim \ \frac {\ell^2}{|\ln \lambda|} %
    && \text{and} && %
    \int_\TTT |\nabla \psi|^2 \ \lesssim \ \frac {\ell^2}{a^2 |\ln
      \lambda|}. \label{3-est}
  \end{align}
  For the reader's convenience, we show the first estimate in
  \eqref{3-est}, the proof of the second inequality proceeds
  similarly. Since $\psi = 1$ on $\overline S$, and since $|\partial
  \overline S^t| = {d \over dt} |\overline S^t|$, we have
  \begin{align*}
    \int_T \psi \ %
    &= \ |\overline S| + \int_0^\infty \phi(t) \ |\partial \bar S^t| \
    dt %
    =  \ - \int_r^a \phi'(t) \ |\bar S^t| \ dt \\
    &\upref{ss-2}\lesssim \frac {\lambda \ell^2}{r^2 |\ln a/r|}
    \int_{0}^a t \ dt \ \lesssim \ \frac {\lambda a^2 \ell^2}{r^2 |\ln
      \lambda|} %
    \lupref{def-rl}= \ \frac {\ell^2}{|\ln \lambda|}.
  \end{align*}

  {\it Step 5: Proof of the lower bound.} We are ready to give the
  proof of the lower bound. It is based on application of Lemma
  \ref{lem-transition} and on a duality argument, using the test
  function $\psi$. We claim that
  \begin{align} \label{elise} %
   \lambda \ell^2 \ \lesssim \ - \ \int_A u_1 \psi.
  \end{align}
  Since $\psi = 1$ and $u_1 \leq -\frac 14$ in $S$, and since $\psi
  \geq 0$ and $u_1 \leq \lambda$ in $A$, it follows for $\lambda \ll
  1$ that
  \begin{align*} %
    - \ \int_A u_1 \psi \ = \ - \ \int_{S \cap \overline S} u_1 \psi -
    \ \int_{A \backslash (S \cap \overline S)} u_1 \psi \ %
    \geq \ \frac{1}{4} |S \cap \overline S| - \lambda \ \int_A \psi \
    \lupupref{S-set}{ss-1}\gtrsim \ \lambda \ell^2.
  \end{align*}
  Application of Lemma \ref{lem-transition} then yields
  \begin{align*} %
    \lambda \ell^2 \ \ &\upref{elise}\lesssim \ - \ \int_A u_1 \psi \
    \lesssim \ E^{\frac 12}[\uu] \left(\frac 1{\gamma^{\frac 12}}
      \nltL{\nabla \psi}{{\TTT}} + \nltL{\psi}{{\TTT}}
    \right)\lupref{3-est}\lesssim E^{\frac 12}[\uu] \ \left(\frac
      {\ell^2}{\gamma a^2 |\ln \lambda|} + \ell^2 \right)^{\frac 12}.
  \end{align*}
  We hence obtain
  \begin{align*}
    E[\uu] \ &\gtrsim \ \ell^2 \min \big \{ \gamma \lambda^2 a^2 |\ln
    \lambda|, \lambda^2 \big \} \ %
    \lupref{def-rl}= \ \ell^2 \min \big \{ \gamma^{\frac 13}
    \eps^{\frac 23} \lambda^{\frac 13} |\ln \lambda|, \ \lambda^2 \big
    \},
  \end{align*}
 contradicting \eqref{E-small}. This concludes the proof of the Proposition.
\end{proof}

\subsection{Sharp interface constructions} \label{ss-upper}

In this section, we present constructions that achieve the optimal
scaling in Theorem \ref{thm-bulk}.  We have two different regimes. For
smaller values of $\lambda$, the optimal scaling of the energy is
achieved by a uniform configuration, while for larger values of
$\lambda$, the optimal scaling is achieved by a self--similar
structure. Note that the constructions in \cite{choksi98} do not yield
the optimal energy in the case of near saturation field. Instead, our
constructions are an adaptation of constructions introduced in
\cite{choksi08} for a model of type-I superconductors. While our
constructions have a similar self--repeating structure to that of
\cite{choksi98}, the definition of the involved functions is different
due to the constraint $|\mm| = \chi_\Omega$ in our model. Our
constructions are also different from those of \cite{choksi08} in the
geometry of the magnetic domains, and are intended to better mimic the
behavior of the minimizers.  The main result of this section is:
\begin{proposition}[Upper bound] \label{prp-2} %
  Suppose that $\lambda \lesssim \gamma^2 |\ln \lambda|^2$.  Then for
  $\eps \ll 1$ and $\lambda \ll 1$, the scaling of the minimal energy
  $E$ is bounded above by
 \begin{equation*} %
   \frac 1{\ell^2} \inf_{\uu \in \AA} E[\uu] \ \lesssim \ \min
   \left \{ \lambda^2, \ \gamma^{\frac 13} \eps^{\frac 23} \lambda |\ln
    \lambda|^{\frac 13} \right\}.
 \end{equation*}
\end{proposition}
Let us remark that the logarithm in the scaling of the energy is a
consequence of the fact that the leading order contribution of the
stray field energy is given by interaction on tangential slices, where
the stray field potential behaves logarithmically.

\medskip

We first note that by choosing the uniform magnetization $\uu =
\lambda \ee_1 \chi_\Omega$, we immediately recover the upper bound
$\frac 1{\ell^2} \inf_{\uu \in \AA} E[\uu] \ \lesssim \ \lambda^2$.
The cross--over to the branched regime occurs at $\lambda \ \sim
\gamma^{\frac 13} \eps^{\frac 23} |\ln \lambda|^{\frac 13}$. It hence
remains to construct an optimal upper bound, if $\lambda$ is larger
than this threshold. In the remaining part of this section, we present
such a construction.  The corresponding estimates are then given in
Section \ref{ss-estimates}, thus completing the proof of Proposition
\ref{prp-2}.

\medskip

Before going into the details of our constructions, however, let us
recall some other constructions that have been proposed in the
literature over the years (for simplicity, we will only discuss the
case $Q \sim 1$). The first estimates of the minimal energy for the
bulk uniaxial ferromagnets go back to the work of Landau and Lifshitz
\cite{landau35} and Kittel \cite{kittel46}. Those constructions were
proposed for zero applied field. In fact, the Landau-Lifshitz
construction cannot be easily extended to the case when the domains
opposing the applied field occupy only a small volume fraction of the
sample. The Kittel construction, on the other hand, can be modified to
account for small volume fraction, resulting in an energy scaling $ E
\ell^{-2} \sim \eps^{\frac 12} \lambda |\ln \lambda|$
\cite{kooy60}. However, since it consists of a ``striped'' domain
pattern, its interfacial energy turns out to be too high at small
$\lambda$. This issue can be addressed by modifying the geometry of
the domains into a lattice of cylindrical ``bubbles'', whose energy my
be estimated as $E \ell^{-2} \sim \eps^{\frac 12} \lambda$
\cite{cape71,druyvesteyn71}.  A comparison of this estimate with the
result of Proposition \ref{prp-1} shows that, although the bubble
construction provides a slight improvement over the stripe
construction, it is highly non-optimal in its scaling behavior with
respect to $\eps$. We note that, in fact, any domain configuration, in
which the domain walls are aligned with the easy axis cannot do better
in terms of energy, and so branching is inevitable for sufficiently
small $\eps$ to reduce energy \cite{choksi99}. On the other hand, if a
tree-like branched domain structure is used (see
\cite[Sec. 4.3]{choksi04} for the construction in the case of type-I
superconductors), it is not difficult to show that the energy will
scale as $E \ell^{-2} \sim \eps^{\frac 23} \lambda^{\frac 23}$. Once again,
comparing this with the result of Proposition \ref{prp-1}, one sees
that, while the considered configuration gives the optimal scaling in
terms of the dependence of the energy on $\eps$, it is highly
non-optimal in terms of $\lambda$. These observations indicate that
the minimizers of $E$ may not have the geometric characteristics of
any of the domain patterns considered above when $\eps \ll 1$ and
$\lambda \ll 1$. In the following we present a construction which
achieves the scaling in Proposition \ref{prp-1}, thus demonstrating
that this scaling is optimal.

\medskip

We begin by fixing the basic geometry. The geometry is an adaption of
a recent self-similar construction for the type-I superconductor model
\cite{choksi08}. However, contrary to the construction in
\cite{choksi08}, our construction includes closure domains. These
closure domains are, in particular, necessary to achieve the optimal
scaling of the energy in the case of soft materials, i.e. $Q \ll
1$. Based on this geometry, we construct two different magnetization
configurations $\uu^{AF}$ and $\uu^{SF}$.  The first configuration
$\uu^{AF}$ avoids anisotropy energy entirely and is optimal for $Q
\gtrsim 1$. The second configuration $\uu^{SF}$ avoids most of the
stray fields and is optimal for $Q \lesssim 1$.

\medskip

\begin{figure}
 \begin{center}
  \includegraphics[width=8cm]{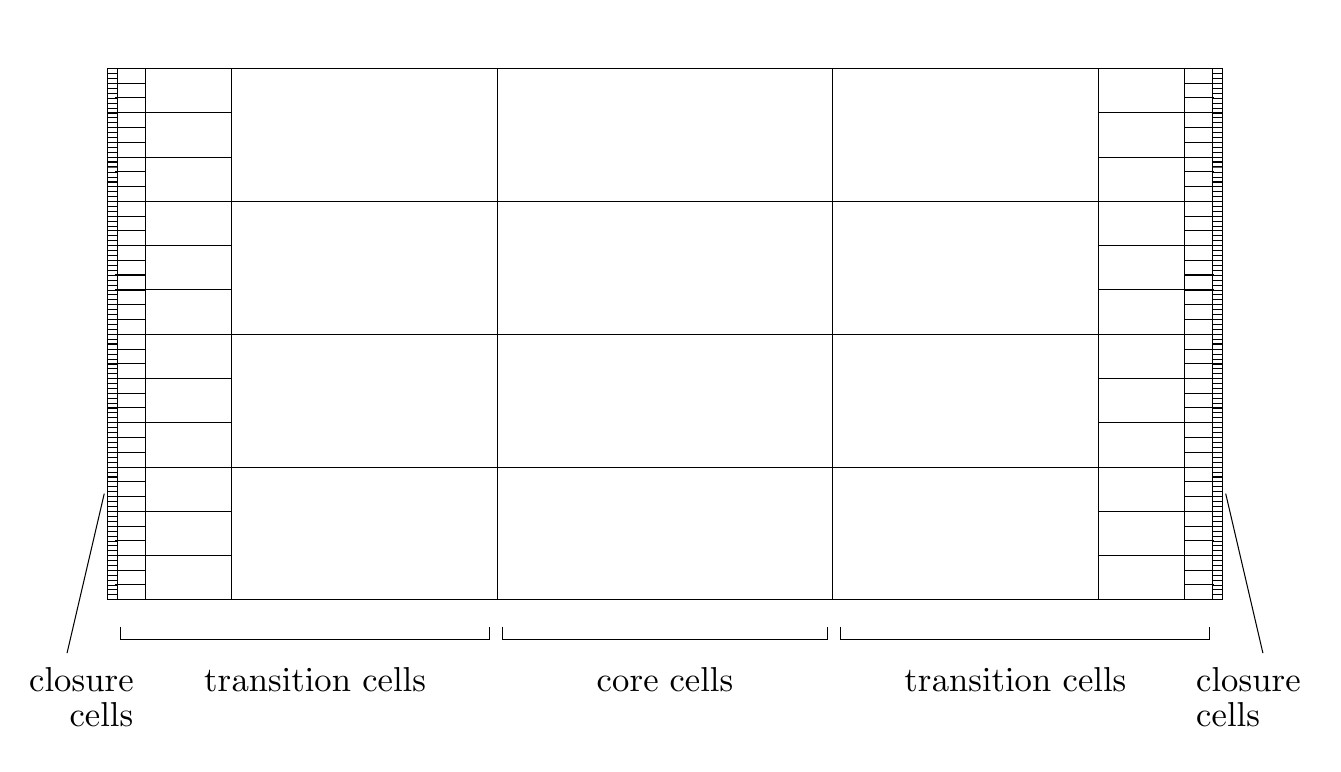}    
 \end{center}
 \caption{A side view (in the $x_1 x_2$-plane) of a sample partition
   containing 16 core cells with 4 generations of refining cells on
   each side. }
 \label{fig-cell2d}
\end{figure}

{\textbf{Sample geometry. \ }} We divide the material plate into three
spatial regions: the core region, the transition region and the
surface region, see Fig. \ref{fig-cell2d}. These regions form 5 layers
symmetrically with respect to the plate's mid-plane. By this symmetry,
it is, therefore, sufficient to describe the constructions only in the
left half of the sample, i.e. for $0 \leq x_1 \leq \frac12$.

\medskip

The core region is partitioned into equal rectangular cells with
height $h_0$ in the normal direction ($x_1$-direction) and length
$a_0$ in both tangential directions. These cells are adjacent on the
left to a system of $M$ layers of self-similar cells in the transition
and surface regions that refine from the core region towards the
boundary (see Fig.  \ref{fig-cell2d}).  Each generation of cells is
described by its height $h_j$ in the normal direction and its
extension $a_j$ in both tangential directions, with $j = 1, \ldots,
M$. We also define a parameter $r_j = \lambda^{\frac 12} a_j / \sqrt{2
  \pi}$, which will be the maximum needle radius in the $j$-th
generation of cells. In every generation, the width of the cells
decreases by a factor $3$, i.e.
\begin{align} \label{width} %
a_{j+1} \ = \ \frac{a_j}3 && \text{and} && r_{j+1} \ = \ \frac{r_j}3.
\end{align}
In particular, the number of cells is multiplied by a factor $9$ in
each new generation.  The algorithm is terminated after $M$
iterations.  We will specify $a_1$, $\{h_j\}$, and $M$ in the
sequel. In particular, these parameters will be chosen, such that the
union of all cells exactly covers $\Omega$, i.e.
\begin{align} \label{total-height} %
  h_0 + 2 \sum_{j=1}^M h_j \ = \ 1.
\end{align}
We differentiate between core cells (which for simplicity of notation
we identify with generation $0$), transition cells (generations $1$,
\ldots, $M-1$) and closure cells (generation $M$).

\medskip

\begin{figure}
\centering
\includegraphics[width=5.5cm]{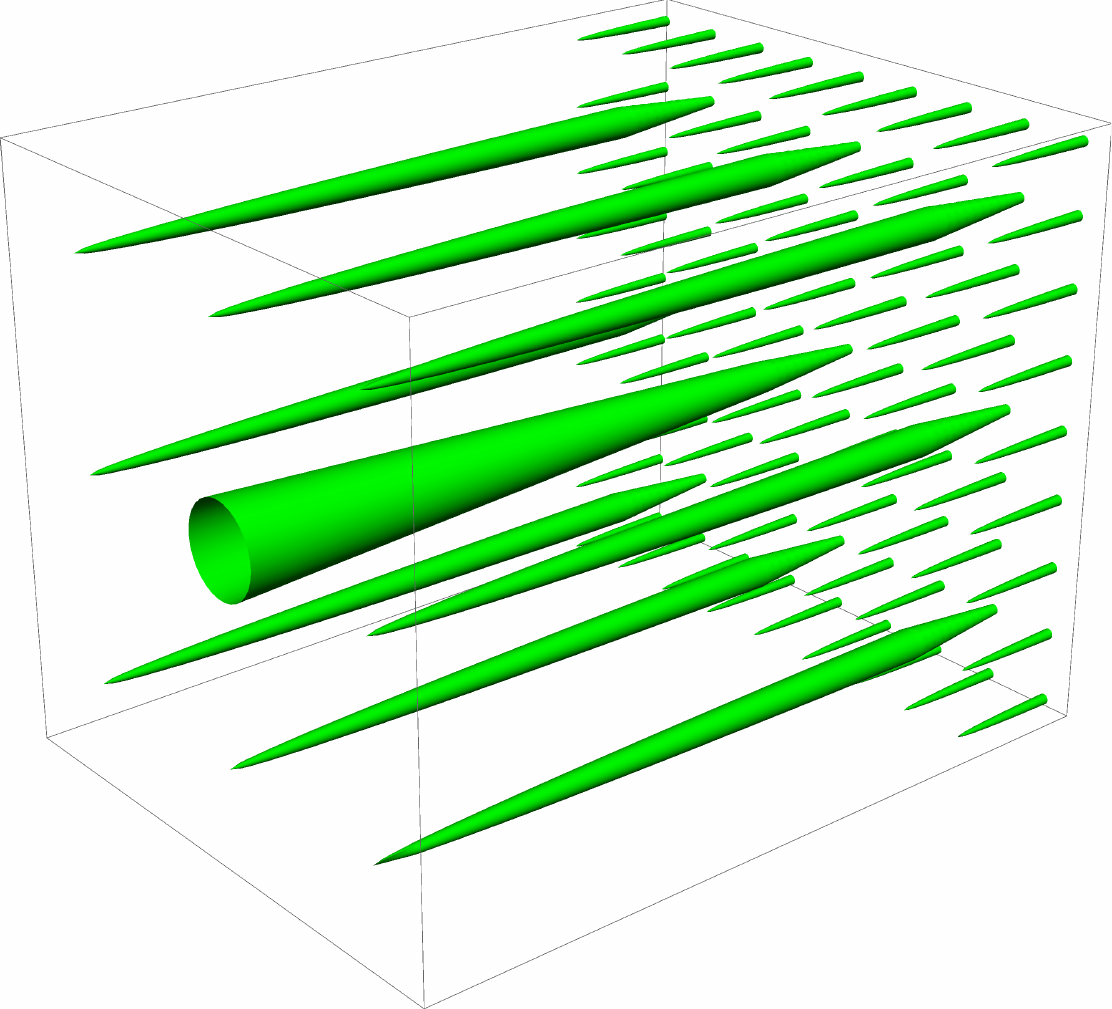}
\caption{Refinement towards the sample boundary (two generations are
  shown).}
\label{fig-refinement}
\end{figure}

Globally, the geometry of the construction consists of a collection of needles rescaled to fit
into the collection of cells just constructed, refining in the direction of the boundary, see
Fig.  \ref{fig-refinement}. We capture the region occupied by the needles by the characteristic
function $\chi \in BV(\R^3, \{ 0, 1\})$, which will be specified in the sequel.  We thus give
the definition of $\chi$ on a rescaled cell
\begin{align} \label{Z} %
  Z \ = \ [0,h] \times K, \qquad \qquad K = \left[-\textstyle \frac a2, \textstyle \frac a2
  \right]^2,
\end{align}
with height $h$ and width $a$. We furthermore denote the tangential boundary of the cell by
$\partial_\perp Z := [0,h] \times \partial K$.  The corresponding ``maximum needle radius'' $r$
is defined by
\begin{align}
  \label{eq:r}
  r = \left( {\lambda \over 2 \pi} \right)^{\frac{1}{2}} a.
\end{align}
The definition of $\chi$ on any cell with arbitrary extension (in the
left side of the sample) is then given by a rescaling of this cell.

\medskip

\textbf{Geometry of a transition cell. \ } Consider a transition cell
$Z_\mathrm{trns}$ first. The cell geometry is characterized by nine
needles, see Fig.  \ref{fig-cell}(a).  The largest needle is located
in the center of the cell and grows into positive $x_1$--direction,
while the other needles are smaller and grow in the negative
$x_1$-direction.  All needles are axially-symmetric around their
corresponding center lines given by $x_\perp = x_\perp^{(i)}$, $i = 1,
\ldots, 9$. The large needle is located in the center of the cell,
i.e.  $x_\perp^{(1)} = 0$. The radii of the needle cross-sections on
tangential slices are functions of $x_1$.  The radius of the large
needle is denoted by $\rho_1(x_1) := \rho_+(x_1)$.  The radii of the 8
small needles are given by $\rho_i(x_1) := \rho_-(x_1)$ for
$i=2,\ldots 9$.  The characteristic function $\chi$ is defined by
\begin{align} \label{def-chi-int}
  \chi(x_1, x_\perp) \ := \ \sum_{i=1}^9 H(\rho_i(x_1) - |x_\perp -
 x_\perp^{(i)}|) %
  && \text{for } (x_1, x_\perp) \in Z_\mathrm{trns},
\end{align}
where $H$ is the Heaviside function, i.e., $H(s) = 1$ for $s > 0$ and
$H(s) = 0$ for $s \leq 0$. It remains to specify the radii $\rho_\pm$
for the large and small needles.

\medskip

\begin{figure}
  \hspace{1.7cm} {\bf a)} \hspace{5.7cm} {\bf b)}
  \vspace{-4mm}
  \begin{center}
   \includegraphics[width=4.5cm]{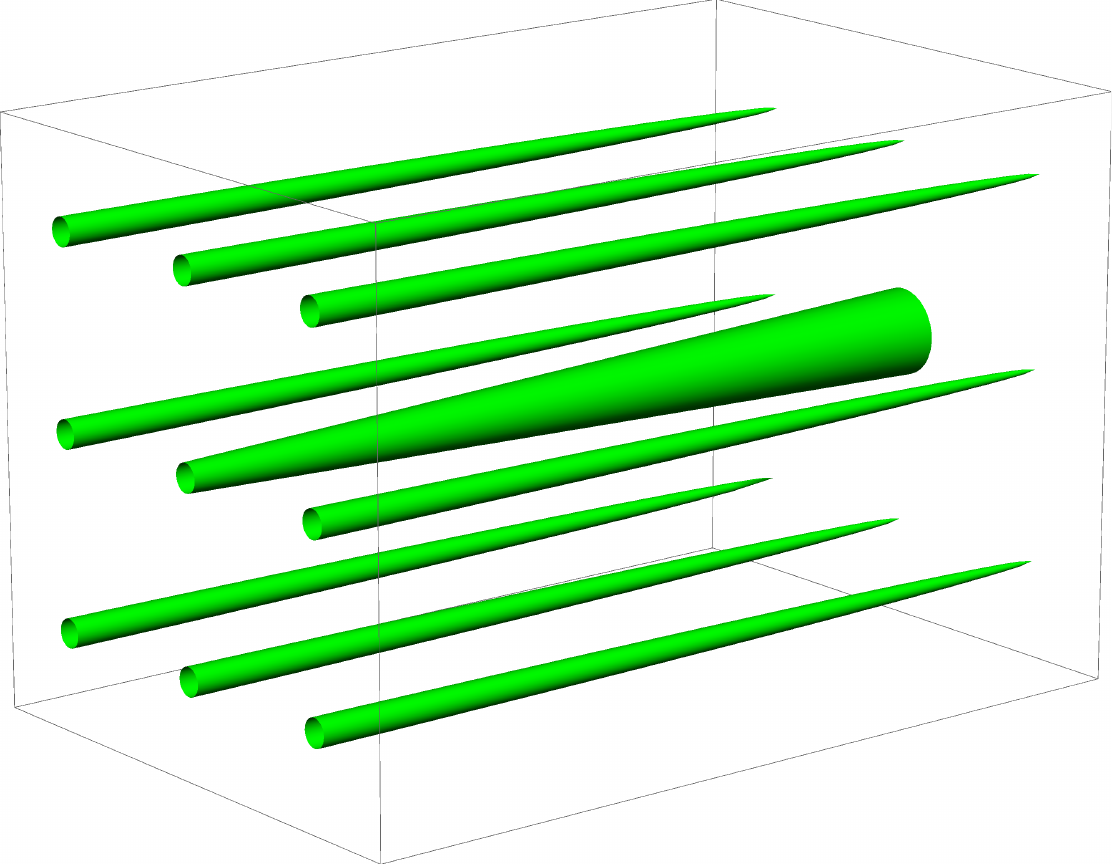} \hspace{1.5cm}
   \includegraphics[width=4.5cm]{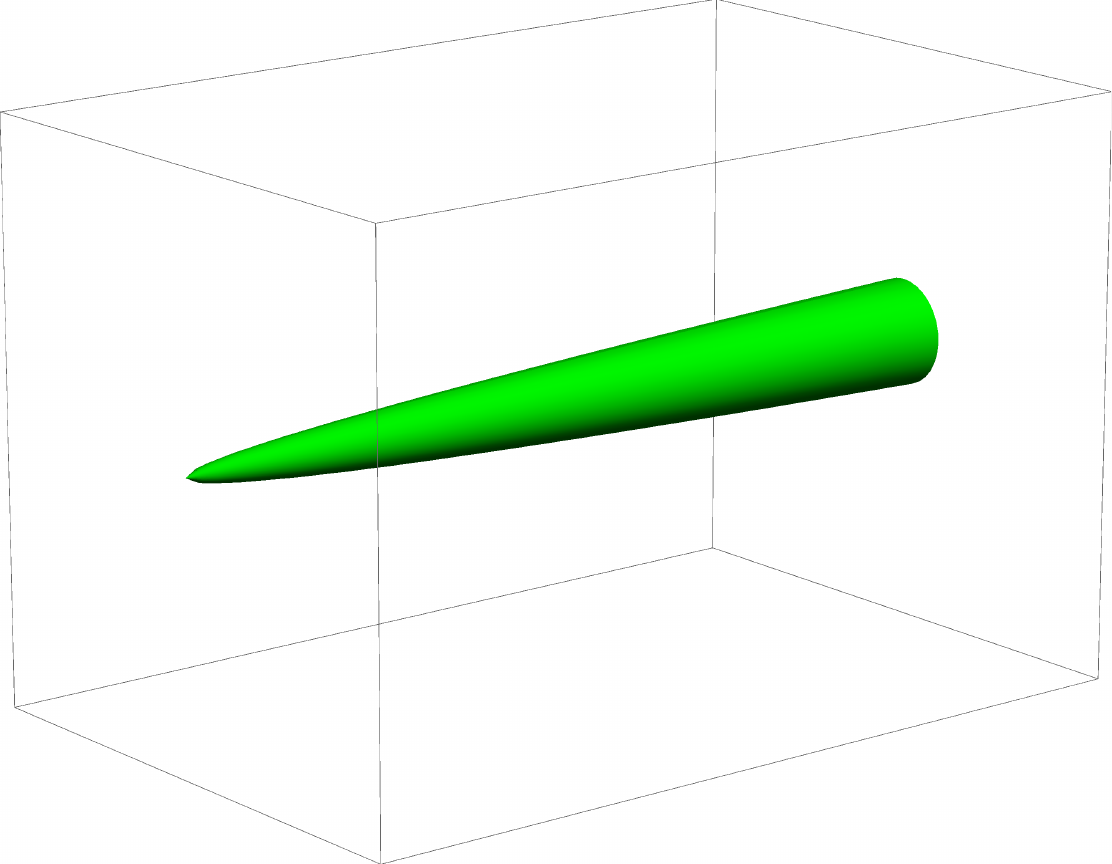}
   \end{center}
   \caption{Geometry of a unit cell: a) transition cell; b) closure
     cell.}
 \label{fig-cell}
\end{figure}

At the tangential faces, i.e. at $x_1 = 0$ and $x_1 = h$, of the cell
the needle radii are defined by
\begin{align*}
  \rho_+(0) = \rho_-(0) = \tfrac13 r, && \rho_-(h) = 0, \quad
  \rho_+(h) = r,
\end{align*}
respectively. This means that at $x_1 = 0$ all needles have the same
radius. At $x_1 = h$, the radii of the small needles are zero. We
choose $\rho_\pm$, such that throughout the cell, the cross-sectional
area of the needles is constant, i.e.
\begin{align} \label{constant-area} %
  \pi \rho_+^2(x_1) + 8 \pi \rho_-^2(x_1) \ = \ \tfrac12 \lambda a^2.
\end{align}
To avoid further complicating the constructions, we assume that the
profile of the small needles is conical at the tip:
\begin{align} \label{needle-tip} %
   h \displaystyle\rho_-(x_1) \ = \ r (h-x_1), %
   &&\displaystyle\text{for $(1 - \alpha) h \leq x_1 \leq h$},
\end{align}
where $0 < \alpha \ll 1$. The precise value of $\alpha$ is
inconsequential and, in fact, for the zero anisotropy configuration
$\uu^{AF}$ could even be taken to be zero. In view of
\eqref{constant-area}, equation (\ref{needle-tip}) defines
$\rho_+(x_1)$ for $(1 - \alpha) h \leq x_1 \leq h$.  Furthermore, on
most of the cell, i.e. for $0 \leq x_1 \leq (1 - \alpha) h$, we choose
$\rho_+$ to be the linear interpolation connecting the values of
$\rho_+$ at $x_1 = 0$ and $x_1 = (1 - \alpha) h$. In turn, $\rho_-$ is
defined by \eqref{constant-area}. Note that as a consequence of the
conical profile at the tip of the small needles (see
\eqref{needle-tip}), it follows that
\begin{align} \label{slender} %
 |\rho_+'(x_1)|, \ |\rho_-'(x_1)| \ \lesssim \ \frac rh && \text{for
    all $x_1 \in [0,h]$}.
\end{align}

\medskip

\textbf{Geometry of a closure cell. \ } In any closure domain cell
$Z_{\rm cls}$, there is only one large needle along the center of the
cell, see \ref{fig-cell}(b). The radius $\rho_{\rm cls}(x_1)$
satisfies the conditions $\rho_{\rm cls}(0) = 0$ and $\rho_{\rm
  cls}(h) = r$. Analogously to the above construction for the
transition cell, we choose a needle shape with a conical tip.  Define
\begin{align}
  \label{chicls}
  \chi(x_1, x_\perp) \ := \ H(\rho_{\rm cls}(x_1) - |x_\perp|) &&
  \text{for } x \in Z_{\rm cls}.
\end{align}
In particular, we also have
\begin{align} \label{slender-cls} %
 |\rho_{\rm cls}'(x_1)| \ \lesssim \ \frac rh && \text{for all $x_1
    \in [0,h]$}.
\end{align}

\medskip

\textbf{Geometry of a core cell. \ } In each core cell
$Z_\mathrm{core}$ the function $\chi$ is assumed to be a
characteristic function of a straight cylinder with the radius equal
to that of the needle on the adjacent side of the transition cell. The
overall geometry of the magnetization pattern for one core cell and 3
refining generations is presented in Fig. \ref{fig-needles}. In the
sequel, we give two different magnetization configurations, $\uu^{AF}$
and $\uu^{SF}$, based on the geometry described above.

\medskip

\textbf{The magnetization $\uu^{AF}$. \ } We first define the
anisotropy-free configuration $\uu^{AF}$ by
\begin{align} \label{def-af} %
  \uu^{AF}(x) \ := \ \bigl(\lambda - 2 \chi(x) \bigr) \ee_1
  \qquad\qquad& \text{in } Z_{\rm core}, Z_{\rm trns}, Z_{\rm cls},
\end{align}
and zero outside $\Omega$. Note that the stray field of $\uu^{AF}$ is
created by surface charges on the needle interfaces and at the sample
surface. We next define an auxiliary field $\tilde \vv^{AF}$. In every
transition cell $Z_{\rm trns}$, we define
\begin{align} \label{def-taf} %
\tilde \vv^{AF}  := \ -\nabla_\perp \phi && \text{in } Z_{\rm trns},
\end{align}
where $\phi$ is a solution of
\begin{align} \label{phi-int} %
  \Delta_\perp \phi \ = \ \nabla \cdot \uu^{AF} \quad \text{in }
 Z_{\rm trns} %
  &&\text{ and }&&
  \partial_{\nnu_{\perp}} \phi \ = \ 0 \quad \text{on } \partial_\perp
  Z_{\rm trns}.
\end{align}
Note that $\vv^{AF}$ is uniquely defined, since by \eqref{def-chi-int}
and \eqref{constant-area} we have
\begin{align}
  \label{eq:29}
  \int_K u_1^{AF} (x_1, x_\perp ) \, d x_\perp = 0 && \forall x_1 \in
  [0, h],
\end{align}
implying the solvability condition for (\ref{phi-int}), in view of the
fact that $\uu^{AF} = u_1^{AF} \ee_1$.  Also note that the
corresponding field $\tilde \vv^{AF}$ is the approximation of the
stray field assuming that magnetostatic interactions in tangential
slices are dominant.

\medskip In closure domain cells, (\ref{eq:29}) does not hold any
more. Hence, in this case  we define $\tilde \vv^{AF}$ by
\begin{align} \label{def-vaf} %
  \tilde \vv^{AF} \ = \ -\nabla_\perp \phi_1 - \partial_1 \phi_2 \
 \ee_1 && \text{in }  Z_{\rm cls},
\end{align}
where $\nabla_\perp \phi_1$ approximates stray field interaction in
tangential directions, while $\partial_1 \phi_2$ approximates stray
field interaction in the normal direction. We define $\phi_1$ as a
solution of
\begin{align} \label{def-pe} %
  \Delta_\perp \phi_1 \ = \ \nabla \cdot \uu^{AF} - \textstyle \frac
  1{a^2} \int_K \nabla \cdot \uu^{AF}(x_1, \hat x_\perp) \, d\hat
 x_\perp \quad \text{in } Z_{\rm cls} %
  && \text{and} &&
    \partial_{\nu_\perp}\phi_1 \ = \ 0 \quad \text{on
    } \partial_\perp Z_{\rm cls}. 
\end{align}
The function $\phi_2$ is defined as a solution of
\begin{align} \label{def-p2} %
  \partial_1^2 \phi_2(x_1) = \textstyle \frac 1{a^2} \int_K \nabla
  \cdot \uu^{AF}(x_1, \hat x_\perp) \, d \hat x_\perp \quad \text{in }
 Z_{\rm cls} %
  && \text{and} && %
  \partial_1 \phi_2 = 0 \quad \text{on } \{x_1 = 0 \} \times K.
\end{align}
Similarly to \eqref{phi-int}, up to a constant, problems
\eqref{def-pe} and \eqref{def-p2} are indeed uniquely solvable for
$\phi_1$ and $\phi_2$. Finally, we set $\tilde \vv^{AF} = 0$ in the
core cells $Z_{\rm core}$ and outside $\Omega$.

\medskip

\textbf{The magnetization $\uu^{SF}$. \ } Here we construct an
approximately stray field-free magnetization $\uu^{SF}$. Although it
would be natural to consider $\uu = \uu^{AF} + \tilde \vv^{AF}$ as a
trial function, in view of the fact that $\nabla \cdot \uu = 0$ in
this case, this function is not admissible, i.e. it does not belong to
$\AA$. For this reason, using the construction of $\uu^{AF}$ and
$\tilde \vv^{AF}$ above, we define an auxiliary function $\tilde
\vv^{SF}$:
\begin{align} \label{def-tsf} %
  \tilde \vv^{SF}(x) := \ \Bigl(1 - \sqrt{1 - |\tilde \vv^{AF}(x)|^2}\
  \Bigr) (\uu^{AF}(x) + (1 - \lambda) \chi_\Omega(x) \ee_1) &&
  \text{in }
  Z_{\rm trns}, \\
  \label{def-tsfc}
  \tilde \vv^{SF}(x) := \ \Bigl(1 - \sqrt{1 - |\tilde
    \vv^{AF}_\perp(x)|^2}\ \Bigr) (\uu^{AF}(x) + (1 - \lambda)
  \chi_\Omega(x) \ee_1) + \tilde v^{AF}_1(x) \ee_1 && \text{in }
  Z_{\rm cls}.
\end{align}
As will be shown in Sec. \ref{ss-estimates}, this definition is
well-posed, since in our construction $|\tilde \vv^{AF}| \ll 1$. We,
therefore, set
\begin{align} \label{def-sf} %
  \uu^{SF} \ := \ \uu^{AF} + \tilde \vv^{AF} - \tilde \vv^{SF}.
\end{align}
It can be easily checked that by our definition of $\tilde \vv^{SF}$,
we have $\uu^{SF} \in \AA$. Furthermore, $\uu^{SF}$ is constructed to
have small stray field (see Sec. \ref{ss-estimates}).

\medskip

\textbf{Localization of the stray field. \ } Note that our
constructions are such that
\begin{align} \label{tas} %
  \nabla \cdot \uu^{AF} \ + \ \nabla \cdot \tilde \vv^{AF} = 0&&
  \text{and} && %
  \nabla \cdot \uu^{SF} \ + \ \nabla \cdot \tilde \vv^{SF} = 0,
\end{align}
in $\mathbb R \times \TTT$.  We will use this information to localize
the estimates for the stray field energy (for the original idea in the
context of ferromagnets, see \cite{choksi98}). Let us note that for
every vector field $\uu$ and its stray field $\vv$ (in the sense of
\eqref{def-v}), we have
\begin{align} \label{h<m} %
  \int_{\R \times \TTT} |\vv|^2 \ = \ \inf_{\tilde \vv} \int_{\R
    \times \TTT} |\tilde \vv|^2,
\end{align}
where the infimum is taken over all fields $\tilde \vv \in L^2(\R^3;
\R^3)$ satisfying
\begin{align} \label{st} %
\nabla \cdot \tilde \vv + \nabla \cdot \uu \ = \ 0
\end{align}
distributionally. This motivates to define for any $\uu \in \AA$ and
for any $\tilde \vv \in L^2(\R \times \TTT; \R^3)$, the energy
\begin{align} \label{tilde-energy} %
  \tilde E[\uu, \tilde \vv] = \eps \int_{\Omega}|\nabla u_1^\delta| +
  Q \int_{\Omega} |\uu_\perp|^2 + \int_{\mathbb R \times \TTT} |\tilde
  \vv|^2.
\end{align}
Hence we have $E[\uu] \ \leq \ \tilde E[\uu, \tilde \vv]$, whenever
$\tilde\vv$ satisfies \eqref{st}.  In view of \eqref{tas} it then
follows that
\begin{align} \label{localize} %
 E[\uu^{AF}] \ \leq \ \tilde E[\uu^{AF}, \tilde \vv^{AF}], %
 && E[\uu^{SF}] \ \leq \ \tilde E[\uu^{SF}, \tilde \vv^{SF}].
\end{align}
The advantage of the quantity $\tilde E$ is that it is local in both
of its parameters. We hence define the restriction of $\tilde E$ on
any set $A \subseteq \R \times \TTT$ by
\begin{align*}
  \tilde E_{|A}[\uu, \tilde \vv] = \eps \int_{A \cap \Omega}|\nabla
  u_1| + Q \int_{A \cap \Omega} |\uu_\perp|^2 + \int_{A} |\tilde
  \vv|^2.
\end{align*}

\subsection{Estimates} \label{ss-estimates} %

In this section, we give the estimates corresponding to the branched
needle construction described in the previous section, thus completing
the proof of Proposition \ref{prp-2}. We will show that
\begin{lemma}[Needles] \label{lem-needle} %
  For $\gamma^{\frac 13} \eps^{\frac 23} |\ln \eps|^{\frac 13}
  \lesssim \lambda \lesssim \gamma^2 |\ln \lambda|^2$, we have
 \begin{align} \label{the-est} %
   \frac 1{\ell^2} \inf_{\uu \in \AA} E[\uu] \ \lesssim \
   \gamma^{\frac 13} \eps^{\frac 23} \lambda |\ln \lambda|^{\frac 13}.
 \end{align}
\end{lemma}
In view of \eqref{localize}, it is enough to give the estimate
\eqref{the-est} in terms of the localized energy $\tilde E$, defined
in \eqref{tilde-energy}, using the approximate stray fields $\tilde
\vv^{AF}$ and $\tilde \vv^{SF}$ defined in the previous section.
Before giving the proof of Lemma \ref{lem-needle}, we estimate the
restriction of $\tilde E$ onto a single transition or closure domain
cell. We note that as long as $h_0$ is bounded away from 1, the
estimates for the core cells trivially result in the same upper bounds
as for the transition cells. Therefore, in the following we do not
include explicit arguments for the core cells.

\medskip

We first give the estimates for $\uu^{AF}$:
\begin{lemma}[Energy of a transition cell for
  $\uu^{AF}$] \label{lem-int-AF} %
  Consider a transition cell $Z_{\rm trns}$ with height $h$ and
  maximum needle radius $r$. Suppose that the needle is slender, in
  the sense of $r \ll h$.  Then
\begin{align*} %
   \tilde E_{|Z_{\rm trns}}[\uu^{AF}, \tilde \vv^{AF}] \ \lesssim \
   \eps r h + \frac{r^4 |\ln \lambda|}{h},
\end{align*}
where $\uu^{AF}$ is defined in \eqref{def-af} and $\tilde \vv^{AF}$
is given by \eqref{def-taf}.
\end{lemma}
\begin{proof}
  By the definition of $\uu^{AF}$, we immediately get the following
  estimate for the surface energy,
 \begin{align*}
   \eps \int_{Z_{\rm trns}}|\nabla u_1^{\delta,AF}| \ \sim \
   \int_{Z_{\rm trns}} |\nabla \chi| \ \lesssim \ \eps r h,
 \end{align*}
 where we used the assumption $r \ll h$. It remains to give the
 estimate for the stray field part of the energy. In view of
 \eqref{def-af}, we get
 \begin{align}
   \label{divAF}
   \nabla \cdot \uu^{AF} \ \lupref{def-af}= \ \partial_1 u_1^{AF} \ =
   \ -2 \sum_{i=1}^9 \delta(\rho_i(x_1) - |x_\perp - x_\perp^{(i)}|)
   \partial_1 \rho_i(x_1),
 \end{align}
 where $\delta$ is the Dirac $\delta$--function.  We expect a
 logarithmic blow--up near each of the needles. This motivates to
 decompose $\phi$, defined in \eqref{def-taf}, by $\phi \ = \
 \phi^{(0)} + \phi^{(1)}$, where $\phi^{(0)}$ is given by
 \begin{align} \label{psipsi-int} %
   \phi^{(0)}(x_1, x_\perp) \ = \ \sum_{i=1}^9 2 \rho_i(x_1)
   \ \partial_1 \rho_i(x_1) H\big( |x_\perp^{(i)} - x_\perp| -
   \rho_i(x_1) \big) \ln \frac {\rho_i(x_1)}{|x_\perp -
     x_\perp^{(i)}|}
 \end{align}
 and where, as before, $H$ is the Heaviside function.  Let us for the
 moment assume that the leading order contribution to the stray field
 energy is due to $\phi^{(0)}$. In view of \eqref{psipsi-int}, it is
 easy to calculate
 \begin{align} \label{D-cell} %
   \int_{Z_{\rm trns}} |\nabla_\perp \phi^{(0)}|^2 \ \lesssim \ \frac
  {r^4 |\ln \lambda| }h %
   && \text{and} && %
   \int_{Z_{\rm trns}} |\nabla_\perp \phi^{(0)}|^4 \ \lesssim \ \frac
   {r^6}{h^3}.
 \end{align}
 We then get
 \begin{align*}
   \int_{Z_{\rm trns}} |\tilde \vv^{AF}|^2 \ %
   \lupref{def-taf}= \ \int_{Z_{\rm trns}} |\nabla_\perp \phi|^2 \ %
   \lesssim \ \ \int_{Z_{\rm trns}} |\nabla_\perp \phi^{(0)}|^2 \ %
   \upref{D-cell}\lesssim \ \frac {r^4 |\ln \lambda|}h.
 \end{align*}
 It remains to check that the energy contribution related to
 $\phi^{(1)}$ is of lower order.  Indeed, from the definition of
 $\phi$ and $\phi^{(0)}$ it follows that $\phi^{(1)}$ satisfies
 \begin{align*}
  \Delta_\perp \phi^{(1)} \ = \ 0 \quad \text{in $Z_{\rm trns}$} %
   && \text{and} && 
   \partial_{\nu_\perp} \phi^{(1)} \ = \ - \partial_{\nu_\perp}
   \phi^{(0)} \quad \text{on $\partial_\perp Z_{\rm trns}$}.
 \end{align*}
 In particular, at the boundary of the cell we have
 \begin{align} \label{phi1-bdry} %
   |\partial_{\nu_\perp} \phi^{(1)}| \ \lesssim \ \sum_{i=1}^9
   \frac{\rho_i |\partial_1 \rho_i|}{|\rho_i - x_\perp^{i}|} \
   \lupref{slender}\lesssim \ \frac{r^2}{ha} %
   && \text{on $\partial_\perp Z$}.
 \end{align}
 It then follows by standard elliptic estimates that
 \begin{align*}
  \int_{Z_{\rm trns}} |\nabla_\perp \phi^{(1)}|^2 \ %
   \lesssim \ h \int_{\partial_\perp Z_{\rm trans}}   
   |\partial_{\nu_\perp} \phi^{(1)}|^2 %
   \lupref{phi1-bdry}\lesssim \ \frac {r^4}h \ %
   \ll \ \frac {r^4 |\ln \lambda|}h.
 \end{align*}
 This completes the proof of Lemma \ref{lem-int-AF}.
\end{proof}

\begin{lemma}[Energy of a closure domain cell for
 $\uu^{AF}$] %
Consider a closure cell ${Z_{\rm cls}}$ with height $h$ and maximum needle radius $r$.
 Suppose that the needle is slender in the sense of $r \ll h$.  Then
\begin{align} \label{ext-af} %
  \begin{gathered}
    \tilde E_{|{Z_{\rm cls}}}[\uu^{AF}, \tilde \vv^{AF}] \ %
    \lesssim \ \eps r h + \frac{r^4 |\ln \lambda|}h + \lambda r^2 h.
  \end{gathered}
\end{align}
\label{lem-cls-AF}
\end{lemma}
\begin{proof}
  The estimate for the energy in the closure domain cells proceeds
  similarly to that for the transition cells. We decompose $\phi_1 \ =
  \ \phi_1^{(0)} + \phi_1^{(1)}$, where the function $\phi^{(0)}$ is
  given by
 \begin{align*} %
   \phi_1^{(0)}(x_1, x_\perp) \ = \ 2 \rho_{\rm cls}(x_1) \ \partial_1
   \rho_{\rm cls}(x_1) H \big(|x_\perp| - \rho_{\rm cls}(x_1) \big)
   \ln \frac {\rho_{\rm cls}(x_1)}{|x_\perp|}.
\end{align*}
 As in the proof of Lemma \ref{lem-int-AF}, it can be shown that the
 contribution of the energy related to $\phi^{(1)}$ can be
 neglected. Hence, in the following, we only give the estimates for
 the functions $\phi_1^{(0)}$ and $\phi_2$. A straightforward
 calculation then yields that we have the following bound on the
 surface energy
 \begin{align*}
   \eps \int_{Z_{\rm cls}}|\nabla u_1^{\delta,AF}| \ \lesssim \ \eps r
   h,
 \end{align*}
 where we used the assumption $r \ll h$. Similarly to the arguments in
 the previous lemma, one can also show that
 \begin{align} \label{D-cell-1} %
   \int_{Z_{\rm cls}} |\nabla_\perp \phi_1^{(0)}|^2 \ \lesssim \ \frac
   {r^4 |\ln \lambda| }h %
   && \text{and} && %
   \int_{Z_{\rm cls}} |\nabla_\perp \phi_1^{(0)}|^4 \ \lesssim \ \frac
   {r^6}{h^3}.
 \end{align}
 Furthermore, from \eqref{def-p2} and \eqref{divAF} we get
 $|\partial_1 \varphi_2| \lesssim \lambda$, and so so
 \begin{align} \label{D-cell-2} %
   \int_{Z_{\rm cls}} |\partial_1 \phi_2|^2 \ \lesssim \ \lambda r^2 h, %
\end{align}
In view of \eqref{D-cell-1} and \eqref{D-cell-2}, the stray field
energy is estimated by
\begin{align} \label{stray-cls} %
  \int_{Z_{\rm cls}} |\tilde \vv^{AF}|^2 \lupref{def-taf}\lesssim \
  \int_{Z_{\rm cls}} |\nabla_\perp \phi^{(0)}|^2 + \int_{Z_{\rm cls}}
  |\partial_1 \varphi_2|^2 \ \lesssim \ \frac {r^4 |\ln \lambda| }h +
  \lambda r^2 h.
\end{align}
This concludes the proof of Lemma \ref{lem-cls-AF}.
\end{proof}

We next give the estimates for $\uu^{SF}$:
\begin{lemma}[Energy of a transition cell for
  $\uu^{SF}$] \label{lem-int-SF} %
  Consider a transition cell ${Z_{\rm trns}}$ with height $h$ and
  maximum needle radius $r$. Suppose that the needle is slender, in
  the sense of $r \ll h$. Then
  \begin{align} \label{int-sf} %
    \tilde E_{|Z_\mathrm{trns}}[\uu^{SF}, \tilde \vv^{SF}] \ \lesssim
    \ \eps r h + \frac{Q r^4 |\ln \lambda|}{h} + \frac{r^6}{h^3}.
\end{align}
\end{lemma}
\begin{proof}
  We use the notation of the proof of Lemma \ref{lem-int-AF}, in
  particular, we decompose $\phi = \phi^{(0)} + \phi^{(1)}$, with
  $\varphi^{(0)}$ being the dominant term. We first note that by
  construction $|\nabla_\perp \varphi| \lesssim r/h \ll 1$ (see
  \eqref{psipsi-int}), so $\uu^{SF}$ is well-defined in $Z_{trns}$.
  Therefore, the surface energy can be estimated as before:
  \begin{align}
    \label{surfSF}
   \eps \int_{Z_{\rm trns}}|\nabla u_1^{\delta,SF}| = \eps
   \int_{Z_{\rm trns}}|\nabla u_1^{\delta,AF}| \ \lesssim \ \eps r h. 
 \end{align}
 The estimate for the anisotropy energy of $\uu^{SF}$ follows from
 \begin{align} \label{aniSF} %
   \int_{Z_{\rm trns}} |\uu_\perp^{SF}|^2 \ %
   \upref{def-tsf}= \ \int_{Z_{\rm trns}} |\nabla_\perp \phi|^2 \ %
   \lesssim \ \int_{Z_{\rm trns}} |\nabla_\perp \phi^{(0)}|^2 \ %
   \lesssim \ \ \frac {r^4 |\ln \lambda|}h.
 \end{align}
 In order to estimate the stray field energy of $\uu^{SF}$, we note
 that in view of \eqref{def-tsf}, we have $|\tilde \vv^{SF}(x)| \
 \lesssim \ |\nabla_\perp \phi(x)|^2$. Hence
 \begin{align*}
   \int_{Z_{\rm trns}} |\tilde \vv^{SF}|^2 \ %
   \lesssim \ \int_{Z_{\rm trns}} |\nabla_\perp \phi|^4 \ %
   \lesssim \ \int_{Z_{\rm trns}} |\nabla_\perp \phi^{(0)}|^4 \ %
   &\upref{D-cell}\lesssim \ \frac{r^6}{h^3}.
 \end{align*}
 The above estimates together yield \eqref{int-sf}.
\end{proof}
\begin{lemma}[Energy of a closure domain cell for $\uu^{SF}$] \label{lem-cls-SF} %
  Consider a closure cell ${Z_{\rm cls}}$ with height $h$ and maximum needle radius $r$.
  Suppose that the needle is slender in the sense of $r \ll h$.  Then
\begin{align*}
  \begin{gathered}
    \tilde E_{|{Z_{\rm cls}}}[\uu^{SF}, \tilde \vv^{SF}] \ %
    \lesssim \ \eps r h + \frac{Q r^4 |\ln \lambda|}h +
    \frac{r^6}{h^3} + \lambda r^2 h.
  \end{gathered}
\end{align*}
\end{lemma}
\begin{proof}
  As in Lemma \ref{lem-int-SF}, the function $\uu^{SF}$ is
  well-defined, since $|\nabla_\perp \phi_1| \lesssim r/h \ll
  1$. Using the slenderness condition and the fact that by
  construction $|\partial_1 \varphi_2| \lesssim \lambda \ll 1$, we
  again get \eqref{surfSF}.  Similarly, \eqref{aniSF} also holds for
  the anisotropy energy.  Finally, in view of the definitions
  \eqref{def-vaf}, \eqref{def-pe}, \eqref{def-p2} and \eqref{def-sf}
  and in view of \eqref{D-cell-1} and \eqref{D-cell-2}, the stray
  field energy is estimated by
\begin{align*}
  \int_{Z_{\rm cls}} |\tilde \vv^{SF}|^2 \ &\lesssim \ \int_{Z_{\rm
      cls}} |\nabla_\perp \phi_1|^4 + \int_{Z_{\rm cls}} |\partial_1
  \phi_2|^2 \ \lupupref{D-cell-1}{D-cell-2}\lesssim \ \frac {r^6}{h^3}
  + \lambda r^2 h.
\end{align*}
The above estimates together conclude the proof of Lemma
\ref{lem-cls-SF}.
\end{proof}

We are ready to give the proof of Lemma \ref{lem-needle}.

\begin{proof}[Proof of Lemma 3.8] %\ref{lem-needle}]
  In view of \eqref{localize}, it is enough to give an estimate in
  terms of the energy $\tilde E$ instead of $E$. As before, consider
  either a transition or a closure cell $Z$ with dimensions $h, a$ and
  the maximum needle radius $r$. The energy of the configurations
  $\uu^{AF}, \tilde \vv^{AF}$ or $\uu^{SF}, \tilde \vv^{SF}$ are
  estimated in the previous Lemmas. In the following, we will write
  $\uu, \tilde \vv$ for a configuration representing either $\uu^{AF},
  \tilde \vv^{AF}$ or $\uu^{SF}, \tilde \vv^{SF}$. Let us for a moment
  assume that all needles are slender, in the sense of $r \ll h$. In
  view of Lemmas \ref{lem-int-AF}--\ref{lem-cls-SF}, we then have for
  any transition cell $Z_{\rm trns}$ and any closure domain cell
  $Z_{\rm cls}$ with the above dimensions the following estimate:
\begin{align}  \label{ic-1} %
%   \begin{gathered} 
     \tilde E_{|Z_{\rm trns}}[\uu, \tilde \vv] \lesssim \eps r h +
     \frac{\gamma r^4 |\ln
      \lambda|}{h} + \frac{r^6}{h^3}  %
     && \text{and} && %
     \tilde E_{|{Z_{\rm cls}}}[\uu, \tilde \vv] %
     \lesssim \eps r h + \frac{\gamma r^4 |\ln \lambda|}h +
     \frac{r^6}{h^3} + \lambda r^2 h.
%   \end{gathered}
 \end{align}
 Balancing the first two terms in the right-hand sides of \eqref{ic-1}
 yields the optimal height for both transition and closure domain
 cells as a function of the maximum needle radius $r$:
 \begin{align} \label{height} %
  h \ = \  \gamma^{\frac 12} \eps^{-\frac 12} r^{\frac 32} |\ln
     \lambda|^{\frac 12}.
\end{align}
With this choice of $h$ for each cell, the estimates in \eqref{ic-1}
turn into
\begin{align} \label{ic-2} %
  \tilde E_{|Z_{\rm trns}}[\uu, \tilde \vv] \ \lesssim \ %
  \gamma^{\frac 12} \eps^{\frac 12} r^{\frac 52} |\ln \lambda|^{\frac
    12} + \frac{r^6}{h^3} && \text{and} && \tilde E_{|Z_{\rm
      cls}}[\uu, \tilde \vv] \ \lesssim \ %
  \gamma^{\frac 12} \eps^{\frac 12} r^{\frac 52} |\ln \lambda|^{\frac
    12} + \frac{r^6}{h^3} + \lambda r^2 h.
\end{align}
Given any initial height $h_1$ for the first generation of cells, we
use \eqref{height} to correspondingly choose the width of the first
generation of cells. The width of the following generation of cells is
inductively defined by \eqref{width} and \eqref{height}. We terminate
the algorithm after $M$ generations of cells as soon as closure domain
cells are not too expensive in the sense of
\begin{align} \label{termination} %
  \lambda r_M^2 h_M \ \ \lesssim \ \frac{\gamma r_M^4 |\ln
    \lambda|}{h_M}.
\end{align}
We choose the initial height $h_1 = h_1(M)$ such that
\eqref{total-height} is satisfied, i.e.  such that the cells exactly
cover $\Omega$. Since in view of \eqref{width} and \eqref{height},
$h_j$ is a geometric sum, as expected we must have $h_1 \sim 1$
independently of $M$. In view of \eqref{height} and
\eqref{termination}, we then get the following estimate for the needle
radius for the first and last generation of cells
\begin{align} \label{ratio} %
  r_1 \ \sim \ \frac{\eps^{\frac 13}}{\gamma^{\frac 13} |\ln \lambda|^{\frac 13}} %
  && \text{and} && %
  r_M \ \sim \ \frac{\eps}{\lambda}.
\end{align}
Note that the termination criterion \eqref{termination} is equivalent
to $r_M^2 /h_M^2 \ \lesssim \ \frac{\lambda}{\gamma |\ln \lambda|}$.
Since in view of \eqref{width} and \eqref{height} $r_j/h_j$ is
monotonically increasing in $j$, we get
\begin{align} \label{satisfied} %
  \frac{r_j^2}{h_j^2} \ \lesssim \ \frac{\lambda}{\gamma |\ln
    \lambda|} && \text{for all $0 \leq j \leq M$}.
\end{align}
Let us assume for the moment that
\begin{align} \label{consistency} %
  \frac{r_j^6}{h_j^3} \ \lesssim \ \frac{\gamma r_j^4 |\ln
    \lambda|}{h_j} %
  && \text{for all $0 \leq j \leq M$}.
\end{align}
In this case, in view of \eqref{ic-2}, the total energy is estimated by
\begin{align*}
  \frac 1{\ell^2} \tilde E[\uu, \tilde \vv] \ &\lesssim \ \frac
  1{a_1^2} \sum_{j=0}^{\infty} \gamma^{\frac 12} \eps^{\frac 12}
  r_j^{\frac 52} |\ln \lambda|^{\frac 12} \ %
  \lupref{ratio}\lesssim \ \gamma^{\frac 13} \eps^{\frac 23} \lambda
  |\ln \lambda|^{\frac 13}.
\end{align*}
In order to complete the proof, it remains to check the following
three consistency criteria.

\medskip

We first need to verify \eqref{consistency}. Indeed,
\eqref{consistency} follows from \eqref{satisfied} and our second
assumption on $\lambda$.  Secondly, we need to check that our
algorithm allows for at least one generation of cells, i.e. $M \gtrsim
1$. In order to see this, we note that $M \gtrsim 1$ is equivalent to
$3^M \gtrsim 1$. By our first assumption on $\lambda$, we have
\begin{align*}
  3^M \ \lupref{width}= \ \frac{r_1}{r_M} \ %
  \lupref{ratio}\sim \ \frac{\lambda}{\gamma^{\frac 13} \eps^{\frac
      23} |\ln \lambda|^{\frac 13}} \ %
  \gtrsim \ \frac{\lambda}{\gamma^{\frac 13} \eps^{\frac 23} |\ln
    \eps|^{\frac 13}} \ %
  \gtrsim 1.
\end{align*}
Finally, we need to check that the cells are indeed slender in the
sense of $r_j \ll h_j$.  Indeed, this follows from the second of
\eqref{satisfied}, together with the second assumption on $\lambda$.
This concludes the proof of Lemma \ref{lem-needle}.
\end{proof}

\subsection{Constructions for the full energy} \label{ss-diffuse}

In this section, we give the proof of Theorem \ref{thm-diffuse}. In
order to do so, it remains to give a diffuse interface version of the
upper constructions in Section \ref{ss-upper}.

\medskip

We first consider the case of a hard material, i.e. $Q \gtrsim
1$. Hence, we will construct a diffuse interface version $\tilde
\uu^{AF}$ of the magnetization $\uu^{AF}$. It is enough to show the
construction for a single needle. For simplicity, consider a closure
cell $Z_\mathrm{cls}$ with height $h$ and width $a$. We recall that
for sharp interfaces, we defined $\uu^{AF}$ by \eqref{def-af}, where
the shape of the needle is described by the characteristic function
$\chi$ in \eqref{chicls}.  Now, we define the transition layer
\begin{align}
  \label{eq:3}
  \SSS \ = \ \big \{ (x_1, x_\perp)\in Z_{\rm cls} \ : \ | \ |
  x_\perp| - \rho_\mathrm{cls}(x_1) - d(x_1) \, | \ \leq \ w(x_1) \big
  \},
\end{align}
where the functions $w(x_1)$ and $d(x_1)$ are the thickness and the
displacement of the diffuse interface, respectively, given by
\begin{align} \label{def-w} %
  w(x_1) := \frac {\eps}{Q r} \, \rho_\mathrm{cls}(x_1), \qquad
  \textstyle d(x_1) := -\rho_\mathrm{cls}(x_1) \left( 1 - \sqrt{1 -
      {\eps^2 \over 3 Q^2 r^2}} \, \right).
\end{align}
We recall that in the present units the quantity $\eps Q^{-1}$ is just
the typical width of Bloch walls, as described in Section
\ref{sec-model}. We also note that by our assumptions $\eps Q^{-1} \ll
r$ and, therefore, $|d(x_1)| \ll w(x_1) \ll \rho_\mathrm{cls}(x_1)$,
for all $x_1 \in (0, h)$ and all cells. Indeed, by our assumptions
$\gamma \gtrsim \lambda^{\frac 12} |\ln \lambda|^{-1} \gg
\lambda$. Therefore, by \eqref{ratio} the inequality holds for the
closure cell, in which $r$ is the smallest.

\medskip

Let the ``mollification of the Heaviside function'' $\tilde H \in
W^{1,\infty}({\R})$ be given by $\tilde H(t) = 0$ in $(-\infty,-1]$,
$\tilde H(t) = \tfrac12 (t+1) \in [-1,1]$ and $\tilde H(t) = 1$ in
$[1, \infty)$. Furthermore let $\tilde H_R(t) := \tilde H(t/R)$.
Analogously to \eqref{chicls}, we define
\begin{align*}
  \tilde \chi(x_1, x_\perp) \ := \ \tilde
 H_{w(x_1)}(\rho_\mathrm{cls}(x_1) + d(x_1)- |x_\perp|) && \text{for
  } x \in Z_\mathrm{cls},
\end{align*}
We then define $\tilde \uu^{AF}$ as follows. First we set 
\begin{align*}
  \tilde u_1^{AF} \ = \ \lambda - 2 \tilde \chi(x) \ \text{ in }
  Z_\mathrm{cls} %
  && \text{ and } && \uu_\perp \ := \ 0 \ \text{ outside of $\SSS$}.
\end{align*}
Inside the transition layer, we define $\tilde \uu_\perp^{AF}$, such
that $|\tilde \uu_\perp^{AF}|$ ensures that $\tilde \uu \in \AA$, and
that the vectors $\ee_1$, $x_\perp$, and $\tilde \uu_\perp^{AF}$ form
a right-handed triplet. A straightforward calculation shows that our
choice of $d(x_1)$ ensures charge neutrality on every slice, i.e.
$\int_K \tilde u_1^{AF} (x_1, x_\perp) \, d x_\perp = 0$ for all $x_1
\in [0,h]$.

\medskip

In the above construction the exchange and anisotropy energy are
supported only in the transition layer $\SSS$. One easily gets that
\begin{align*}
  \int_\SSS \left( \frac{\eps^2}{4 Q} |\nabla \tilde \uu^{AF}|^2 + Q
    |\tilde \uu_\perp^{AF}|^2 \right) \ %
  \lesssim \ \int_0^h \left( \frac {\eps^2 \rho_\mathrm{cls}(x_1)}{4 Q
      w(x_1)} + Q \rho_\mathrm{cls}(x_1) w(x_1) \right) dx_1 \ \sim \
  \eps r h.
\end{align*}
In estimating the stray field $\tilde{\vv}^{AF}$ of $\tilde \uu^{AF}$,
we can follow the same arguments as for $\uu^{AF}$, modifying the
definition of $\varphi_1^{(0)}$ to be the radially-symmetric potential
due to charges $\nabla \cdot \tilde \uu^{AF}$ rather than $\nabla
\cdot \uu^{AF}$. It is easy to see that all the estimates remain
unchanged. Comparing the above estimates with those in \eqref{ext-af},
it follows that the localized diffuse interface energy $\EEE[\tilde
\uu^{AF}, \tilde{\vv}^{AF}]$ for a single cell, based on the
construction $\tilde \uu^{AF}$ is not larger in terms of scaling than
the localized sharp interface energy $E[\uu^{AF}, \tilde \vv^{AF}]$ of
the optimal sharp interface construction in Section \ref{ss-upper}.

\medskip

The function $\tilde \uu^{AF}$ can be defined throughout $\Omega$ by
applying the above construction to every needle in the self-similar
geometry described in Section \ref{ss-upper}. The corresponding
estimate for $\tilde \uu^{AF}$ follows. Lastly, we note that the
construction in the case of a soft material ($Q \lesssim 1$) proceeds
analogously. This concludes the proof of Theorem \ref{thm-diffuse}.

\section{Reduced energy} \label{sec-reduced} %

The analysis of Section \ref{sec-bulk} provides the scaling of the
minimal energy in the limit of thick samples or, correspondingly, when
$\eps \to 0$.  While the analysis does not require any assumptions
about the minimizers (the analysis performed is ansatz-free), the
results obtained give only a rough idea about the structure of the
minimizers. It seems natural to expect that the minimizers should look
like the trial functions used in the construction of the upper bounds.
Yet, the precise shape of the domains, as well as the precise
constants in the asymptotic behavior of the minimal energy cannot be
captured by the analysis above, since it does not address the leading
order constant in the scaling of the energy.

\medskip

In this section, under the assumption that the magnetization is mostly
aligned with the easy axis and that the geometry of the minimizers is
slender, we derive a reduced sharp interface energy which should
provide the leading order behavior of energy for $\eps \ll 1$. Our aim
is to reduce energy minimization of $E$ to a two-step process. In the
first step, we fix the ``shape'' of the magnetic domains and construct
the energy-minimizing configuration of the magnetization away from the
domain walls. In the second step, we minimize the obtained, reduced
energy, which depends only on that shape. We note that the heuristic
idea of computing the combined contribution of the magnetostatic and
anisotropy energies away from the domain walls has been known as the
$\mu^*$-method in the physics literature \cite{williams49,hubert}.
Below, we assign a precise mathematical meaning to this idea and
provide its rigorous justification under specific assumptions.
Finally, let us also point out that while in this paper we are
interested in the case of applied field near saturation ($\lambda \ll
1$), we expect that the obtained reduced energy to be valid
independently of the applied field, even in zero applied field
($\lambda = 1$).

\medskip 

We first introduce the characteristic function $\chi$ representing the
shape of the domains where the magnetization vector is not aligned
with the external field,
\begin{align} \label{def-chi} %
  \chi(x) \ = \ 1 \ \text{if \ } m_1(x) < 0, && \chi(x) \ = \ 0 \
  \text{if \ } m_1(x) \geq 0.
\end{align}
The reduced energy which is derived in this section is then given by
\begin{align} \label{eq:21} %
  E_0[\chi] = 2 \eps \int_\Omega |\nabla \chi| + \gamma \int_{\mathbb
    R \times \TTT} \partial_1 (\lambda \chi_\Omega - 2 \chi)
  (-\Delta_Q^{-1})
\partial_1  (\lambda \chi_\Omega - 2 \chi),
\end{align}
where the operator $\Delta_Q$ is defined by
\begin{align*}
 \Delta_Q = \big(\chi_\Omega + \gamma (1 - \chi_\Omega) \big)
 \Delta_\perp + \gamma
\partial^2_1,
\end{align*}
with constant $\gamma$ defined earlier in \eqref{def-gamma}. The
admissible class of functions $\chi$ for $E_0$ is
\begin{align*}
\AA_0 = \big \{ \chi \in BV({\R \times \TTT}; \{0,1\}) : \chi
= 0 ~\mathrm{in}~ (\mathbb R \times \TTT) \backslash \Omega \big \}.
\end{align*}

\medskip

Note that $\Delta_Q \approx \Delta_\perp$ in $\Omega$, when acting on
functions that vary slowly in the easy direction, compared to the
directions normal to the easy axis. There is a second equivalent
formulation for \eqref{eq:21}. Indeed, a straightforward calculation
yields that \eqref{eq:21} can be written as
\begin{align} \label{eq:E00} %
  E_0[\chi] = \lambda^2 \ell^2 + 2 \eps \int_\Omega |\nabla \chi| - 4
  \lambda \int_\Omega \chi + 4 \gamma \int_{\mathbb R \times \TTT}
\partial_1 \chi (-\Delta_Q^{-1}) \partial_1 \chi.
\end{align}
Notice that by lower semicontinuity and coercivity, the minimum of the
energy of $E_0$ is attained in $\AA_0$ (see also \cite[Theorem
1.2]{choksi98}).

\medskip

We first note that up to the leading order constant, the reduced
energy $E_0$ has the same scaling of minimal energy as $E$, i.e., for
$\lambda \lesssim \gamma^2 |\ln \lambda|^2$ and $\ell$ sufficiently
large we have
\begin{equation*}
 \frac 1{\ell^2} \inf_{\chi \in \AA_0} E_0[\chi] \ \sim \ %
 \min \left \{ \lambda^2, \gamma^{\frac 13} \eps^{\frac 23} \lambda |\ln
   \lambda|^{\frac 13} \right\}.
\end{equation*}
It can be checked that this result follows by a slight modification of the proof of Theorem
\ref{thm-bulk} (replacing $u_1$ with $\lambda \chi_\Omega - 2 \chi$).

\medskip

In order to show the asymptotic equivalence of the minimum energies
for $E$ and $E_0$, including the leading order constant, we need to
make an assumption on the magnetization $\mm$. If we assume that the
magnetization vector $\mm$ does not deviate strongly from the easy
axis throughout the sample and, furthermore, that the geometry of the
magnetization configuration is slender, we can show that the minimal
energies for $E_0$ and $E$ essentially agree to the leading order:
\begin{theorem} \label{t:E0} %
  Let $\eps \ll 1$ and $\gamma \gg \delta$, where $0 < \delta \ll 1$
  is the same as in \eqref{energy}. Then
  \begin{enumerate}
  \item[(i)] For every $\uu \in \AA$ satisfying $|\uu_\perp| <
    \delta$ there exists $\chi \in \AA_0$, such that 
   \begin{align*}
      E[\uu] \geq (1 - \delta^{\frac 12}) E_0[\chi].
   \end{align*}
  \item[(ii)] For every $\chi \in \AA_0$, for which the solution
    $\tilde\varphi$ of
    \begin{align} \label{eq:7-2} %
      \Delta_Q \tilde \varphi = \gamma \partial_1 (\lambda \chi_\Omega
      - 2 \chi).
    \end{align}
    satisfies $|\nabla \tilde\varphi| < Q \delta$, there exists $\uu
    \in \AA$, such that
   \begin{align*}
      E_0[\chi] \geq (1 - \delta^{\frac 12}) E[\uu].
   \end{align*}
  \end{enumerate}
\end{theorem}
Theorem \eqref{t:E0} is proved at the end of this section via
Propositions \ref{p:reduced} and \ref{p:reducedup}. Before giving the
proofs for Propositions \ref{p:reduced} and \ref{p:reducedup}, it is
instructive to present a formal derivation of $E_0$ (see also
\cite{williams49}).

\medskip

We note that both anisotropy and the external field favor alignment of
the magnetization with the easy axis. We hence expect that $\mm
\approx \pm \mathbf e_1$ in $\Omega$ and should, therefore, have $u_1
\approx \lambda - 2 \chi$ there. In view of $|\mm| = \chi_\Omega$,
this motivates to write $\uu$ in $\Omega$ in the form
\begin{align*} 
 \uu \ %
  &= \left(\lambda - 1 + (1 - 2 \chi) \sqrt{1 - |\uu_\perp|^2}
  \right) \mathbf e_1 + \uu_\perp \ = \ (\lambda - 2 \chi ) \,
  \mathbf e_1 + \uu_\perp + \OOO(|\uu_\perp|^2).
\end{align*}
For $|\mm_\perp| = |\uu_\perp| \ll 1$, to the leading order
(\ref{energy}) then turns formally into
\begin{align} \label{eq:E0u} %
E[\uu] \simeq 2 \eps \int_\Omega |\nabla \chi| + Q \int_\Omega
|\uu_\perp|^2 + \int_{\mathbb R \times \TTT} |\tilde{\mathbf
  v}|^2.
\end{align}
Here, the function $\tilde{\mathbf v}$ satisfies 
\begin{align}  \label{eq:6} %
\tilde{\mathbf v} = -\nabla \tilde\varphi, \qquad \Delta \tilde\varphi
= \partial_1 
(\lambda \chi_\Omega - 2 \chi) + \nabla_\perp \cdot
\uu_\perp \quad \mathrm{in} \quad \mathbb R \times \TTT.
\end{align}

\medskip

Following our approach, we minimize \eqref{eq:E0u} in two steps. First
we take the minimum with respect to $\uu_\perp$ with $\chi$ fixed. In
the second step we minimize the result with respect to all admissible
characteristic functions $\chi$. It is not difficult to see (see the
proof of Proposition \ref{p:reduced} for details) that for fixed
$\chi$ the sum of the last two terms in (\ref{eq:E0u}) is minimized
when
\begin{align*}  %
\uu_\perp = -Q^{-1} \chi_\Omega \nabla_\perp \tilde \varphi.
\end{align*}
Substituting this relation into (\ref{eq:6}), we find that
\begin{align*}
  \partial_1^2 \tilde \varphi + \Delta_\perp \tilde \varphi %
  &= \ \partial_1 (\lambda \chi_\Omega - 2 \chi) + \nabla_\perp \cdot
 \uu_\perp \ = \
  \partial_1 (\lambda \chi_\Omega - 2 \chi) - Q^{-1} \chi_\Omega
  \Delta_\perp \tilde \varphi,
\end{align*}
which is precisely \eqref{eq:7-2}. Finally, substituting the
expression for $\tilde \varphi$ in \eqref{eq:7-2} into \eqref{eq:E0u},
it then follows that $E[\uu] \simeq E_0[\chi]$, where the
reduced energy $E_0$ is given by \eqref{eq:21}. Here we took into
account that $\Delta_Q$ is an invertible operator. Thus, to the
leading order the energy of minimizers of $E_0$ should coincide with
that of $E$.

\medskip

Observe that the reduced energy $E_0[\chi]$ just derived has a form
very similar to that of the sharp interface energy $E$ in the case of
infinite anisotropy, $Q = \infty$. Indeed, in the latter case the
magnetization vector $\mm$ is restricted to take only two values in
$\Omega$: $\mm = \pm \mathbf e_1$. For such magnetization
configurations the sharp interface energy \eqref{energy} turns into
\begin{align*}
 E^{AF}[\uu] = \eps \int_\Omega |\nabla u^\delta_1 | +
  \int_{\mathbb R \times \TTT} \partial_1 u_1 (-\Delta^{-1})
  \, \partial_1 u_1,
\end{align*}
which coincides with (\ref{eq:21}) for $\gamma = 1$, since in this
case $u_1 = \lambda \chi_\Omega - 2 \chi$. On the other hand, noting
that for slender magnetization configurations we have $-\Delta^{-1}
\simeq -\Delta_Q^{-1} \simeq -\Delta_\perp^{-1}$, one should expect
\begin{align}
  \label{eq:23}
  E^{AF} & \simeq 2 \eps \int_\Omega |\nabla_\perp \chi | +
  \int_{\mathbb R \times \TTT} \partial_1 (\lambda \chi_\Omega - 2
  \chi) (-\Delta_\perp^{-1}) \, \partial_1 (\lambda \chi_\Omega - 2
  \chi),
  \\
  E_0 & \simeq 2 \eps \int_\Omega |\nabla_\perp \chi | + \gamma
  \int_{\mathbb R \times \TTT} \partial_1 (\lambda \chi_\Omega - 2
  \chi) (-\Delta_\perp^{-1}) \, \partial_1 (\lambda \chi_\Omega - 2
  \chi), \label{eq:23a}
\end{align}
where $\chi$ in (\ref{eq:23}) and (\ref{eq:23a}) is given by
(\ref{def-chi}). It is easy to see that up to a multiplicative factor
the expression in (\ref{eq:23a}) coincides with that in (\ref{eq:23})
after rescaling $x_2$ and $x_3$ with $\gamma^{-\frac 13}$. Similarly,
the energy per unit area, according to (\ref{eq:23a}), is
$\gamma^{\frac 13}$ times the expression in (\ref{eq:23}). Thus, the
two energies approximately agree with each other for $Q \gg 1$, and
the minimal energy per unit area is smaller by a factor $Q^{\frac 13}$
for $Q \ll 1$ (see \eqref{def-gamma}). The latter is due to the fact
that in this case the {\em effective} magnetostatic interaction is
weakened by a factor of $Q$, since the stray field is shielded by
small deviations of $\mm$ from the easy direction, creating magnetic
counter-charges. These arguments provide a physical explanation of the
apparently surprising fact that the scalings of the energy of
non-trivial minimizers both in the case of hard and soft materials
agree up to a factor of $Q^{\frac 13}$ (see also
\cite{williams49,hubert}).

\medskip

We now give a rigorous derivation of the relationship between the
reduced energy $E_0$ and the sharp interface energy $E$ under a few
assumptions which appear quite natural physically. Theorem \ref{t:E0}
is an immediate consequence of the two propositions that follow (for
the lower and upper bounds, respectively).  We begin with the analysis
of the lower bound for $E[\uu]$ in terms of $E_0[\chi]$.
\begin{proposition}
 \label{p:reduced}
 Let $\gamma \gg \delta$, where $0 < \delta \ll 1$ is the same as in
\eqref{energy}, let $\uu \in \AA$, and let $|\uu_\perp| <
 \delta$.  Then $E[\uu] \geq (1 - \delta^{\frac 12}) E_0[\chi]$,
 where $\chi$ is given by (\ref{def-chi}).
\end{proposition}

\begin{proof}
 Let us write $\uu = \uu^{(0)} + \uu^{(1)}$ and
  $\varphi = \varphi^{(0)} + \varphi^{(1)}$, where $\varphi$ is
  defined in (\ref{def-v}), and
  \begin{align}
    \label{eq:28}
   \uu^{(0)} = \chi_\Omega \{ (\lambda - 2 \chi) \mathbf e_1 -
    Q^{-1} \nabla_\perp \varphi^{(0)}\}, \qquad \Delta \varphi^{(0)} =
   \nabla \cdot \uu^{(0)}.
  \end{align}
  This is always possible, since $\varphi^{(0)}$ is uniquely solvable
  in terms of $\chi$.  Indeed, eliminating $\uu^{(0)}$ in the second
  equation in (\ref{eq:28}) via the first equation in (\ref{eq:28}),
  one immediately sees that $\varphi^{(0)}$ solves the same equation
  as $\tilde\varphi$ in (\ref{eq:7-2}). Then, after a straightforward
  computation, using \eqref{def-v} and \eqref{eq:28}, the bulk part of
  energy $E_\mathrm{bulk}[\uu] = Q \int_\Omega |\uu_\perp|^2 +
  \int_{\mathbb R \times \TTT} |\nabla \varphi|^2$ can be written as
 \begin{align*}
    E_\mathrm{bulk}[\uu] %
    &= Q \int_\Omega \left| \uu^{(1)}_\perp - Q^{-1} \nabla_\perp
      \varphi^{(0)} \right|^2 + \int_{\mathbb R \times \TTT} \left|
      \nabla (\varphi^{(0)} +
      \varphi^{(1)}) \right|^2 \notag \\
    &= \int_\Omega \left( Q \left| \uu^{(1)}_\perp \right|^2 + Q^{-1}
      \left| \nabla_\perp \varphi^{(0)} \right|^2 + 2 \partial_1
      \varphi^{(0)} \, u_1^{(1)} \right) + \int_{\mathbb R \times
      \TTT} \left( \left| \nabla \varphi^{(0)} \right|^2 + \left|
        \nabla \varphi^{(1)} \right|^2 \right).
 \end{align*}
 Using Cauchy-Schwarz inequality, we obtain
 \begin{align} \label{eq:11}
    E_\mathrm{bulk}[\uu] \geq \int_{\mathbb R \times \TTT} (1 +
    Q^{-1} \chi_\Omega) |\nabla_\perp \varphi^{(0)}|^2 +
    \int_{\mathbb R \times \TTT} |\partial_1 \varphi^{(0)}|^2 
    - 2 \left( \int_\Omega |\partial_1 \varphi^{(0)}|^2
      \int_\Omega |u_1^{(1)}|^2 \right)^{\frac 12}.
  \end{align}

  \medskip

  Now, recalling \eqref{def-uv} and \eqref{def-chi}, one can see that
  in $\Omega$ the angle between the vector $\mm$ and the vector
  $m^{(0)}_1 \ee_1$, where $m^{(0)}_1$ is the first component of the
  vector $\mm^{(0)} = \uu^{(0)} + (1 - \lambda) \ee_1$, does not
  exceed $\frac{\pi}{2}$. Note that by \eqref{eq:28} we have
  $|m^{(0)}_1| = 1$. Therefore, since $|\mm| = 1$ as well, the
  magnitude of the projection of the vector $\uu^{(1)} = \mm -
  \mm^{(0)}$ onto $\ee_1$ does not exceed 1, and $u^{(1)}_1 \ee_1$
  points in the direction opposite to $m^{(0)}_1 \ee_1$. From this we
  conclude that $|u_1^{(1)}| = 1 - \sqrt{1 - |\uu_\perp|^2} \leq
  |\uu_\perp|^2$.  Therefore, by assumption we have $|u_1^{(1)}| \leq
  \delta |\uu_\perp|$.  On the other hand, as can be easily seen, the
  non-local part of $E_0$ equals the sum of the first two terms in
  (\ref{eq:11}). Furthermore, by Cauchy-Schwarz inequality
  \begin{align*}
    2 \left( \int_\Omega \left| \partial_1 \varphi^{(0)} \right|^2
     \int_\Omega \left| \, u_1^{(1)}
      \right|^2 \right)^{\frac 12} %
   \leq 2 \left( E_0[\chi] \delta^2 \int_\Omega |\uu_\perp|^2
    \right)^{\frac 12} %
    \leq 2 \delta Q^{-\frac 12} E_0[\chi]^{\frac 12} E_0[\uu]^{\frac 12}.
 \end{align*}
 Combining the above estimates with the fact that
 \begin{eqnarray*}
    \int_\Omega |\nabla u^\delta_1| = 2 (1 - \delta^2) \int_\Omega
    |\nabla \chi|,
 \end{eqnarray*}
  we arrive at
 \begin{eqnarray*}
    E[\uu] \geq 2 \eps (1 - \delta^2) \int_\Omega |\nabla \chi|
    + E_\mathrm{bulk}[\uu] \ \geq (1 - \delta^2) E_0[\chi] - 2
    \delta Q^{-\frac 12} E ^{\frac 12}[\uu] \, E_0 ^{\frac 12}[\chi],
  \end{eqnarray*}
  which yields the statement. 
\end{proof}

We next give the proof for the upper bound for $E[\uu]$ in terms of
$E_0[\chi]$.

\begin{proposition}
\label{p:reducedup}
Let $\mathbf \chi \in \AA_0$, let $|\nabla_\perp \tilde \varphi| \leq
Q \delta$, where $\tilde \varphi$ is given by \eqref{eq:7-2}, for some
$0 < \delta \ll \gamma$.  Then there exists $\uu \in \AA$ with
$|\uu_\perp| \leq \delta$, such that $E_0[\chi] \geq (1 -
\delta^{\frac 12}) E[\uu]$.
\end{proposition}

\begin{proof}
 Set $\varphi^{(0)} = \tilde \varphi$ and $\uu = \mathbf
  u^{(0)} + \uu^{(1)}$, where $\uu^{(0)} = \chi_\Omega (
  (\lambda - 2 \chi) \mathbf e_1 - Q^{-1} \nabla_\perp
 \varphi^{(0)})$, and $\uu^{(1)} = u^{(1)}_1 \mathbf e_1$
  ensures that $|\uu^{(0)} + \uu^{(1)} + (1 - \lambda)
  \chi_\Omega \mathbf e_1| = \chi_\Omega$. Then $\varphi^{(0)}$ solves
 $\Delta \varphi^{(0)} = \nabla \cdot \uu^{(0)}$, and by
  assumption $|u^{(1)}_1| \leq Q^{-2} |\nabla_\perp
  \varphi^{(0)}|^2$. The proof is then obtained by retracing the
  calculation in the proof of Proposition \ref{p:reduced}, noting that
  in this case $\uu^{(1)}_\perp = 0$. We will only need one extra
  estimate for $\varphi^{(1)}$ solving $\Delta \varphi^{(1)} = \nabla
 \cdot \uu^{(1)}$.

\medskip

Integrating by parts and applying Cauchy-Schwarz inequality, we obtain
\begin{align*}
  \int_{\mathbb R \times \TTT} |\nabla \varphi^{(1)}|^2 &= -
  \int_{\mathbb R \times \TTT} \varphi^{(1)} \Delta \varphi^{(1)} = -
  \int_{\mathbb R \times \TTT} \varphi^{(1)} \partial_1 u^{(1)}_1 =
  \int_\Omega \partial_1 \varphi^{(1)} u^{(1)}_1 \leq \left(
    \int_\Omega |\nabla \varphi^{(1)}|^2 \int_\Omega |u^{(1)}_1|^2
  \right)^{\frac 12}.
\end{align*}
Squaring both sides and using the above estimate for $\left| u^{(1)}_1
\right|$, we find that
\begin{align*}
 \int_\Omega \left| \nabla \varphi^{(1)} \right|^2 \leq Q^{-2}
 \delta^2 \int_\Omega \left| \nabla_\perp \varphi^{(0)} \right|^2 \leq
 \delta^2 Q^{-1} E_0[\chi],
\end{align*}
which completes the proof.
\end{proof}

\section{Transition to non-trivial minimizers}
\label{sec-transition}

As we showed in Theorem \ref{thm-diffuse}, there is a change in the
scaling behavior of the minimum energy due to appearance of
non-trivial minimizers of $\mathcal E$ at $\lambda \sim \gamma^{\frac 13}
\eps^{\frac 23} |\ln \eps|^{\frac 13}$. In this section we analyze the nature of
this transition in more detail. Specifically, we are interested in
locating the precise critical value of $\lambda$ (corresponding to the
critical applied field away from saturation) at which this transition
occurs. We also address the structure of the domain patterns near the
transition point.

\medskip

As a first step, based on an asymptotic study of the reduced energy
$E_0$, we derive an even further reduced energy $E_{00}$. We will
formally show that near the transition point, in rescaled variables,
it is appropriate to consider the energy
\begin{align} \label{eq:37} %
  \bar E_{00}[\bar A] = \int_0^h \left\{ {1 \over 2 \pi} \left( {d
        \bar A \over d \xi} \right)^2 - \bar A + \sqrt{\pi \bar A}
  \right\} d \xi,
\end{align}
where the set of admissible functions is 
\begin{align*}
  \overline \AA \ = \ \big \{ \bar A \in H^1((0,h)) \ : \ \bar A \
  \geq \ 0, \ \bar A(0) = \bar A(h) = 0 \big \}
\end{align*}
The function $\bar A(\xi)$ is simply the rescaled area of the
cross-section of a single needle as a function of the rescaled
coordinate along the needle.  The single parameter $h$ can be
understood as a measure for the effective thickness of the plate. It
is defined by
\begin{align} \label{eq:14} %
  h := {\lambda^{\frac 32} \over \eps \gamma^{\frac 12} \ln^{\frac 12} \bigl(
    \gamma^{\frac 13} \eps^{-\frac 13})}.
\end{align}
A detailed derivation of the reduced energy $\bar E_{00}$ and the
precise definition of the quantities $\bar A$ and $\xi$ in terms of
the original quantities is given in the next subsection.

\medskip

The advantage of the energy $\bar E_{00}$ is that it can be explicitly
minimized and its minimizers can be explicitly computed. We identify
two critical values, denoted as $h_0^*$ and $h_1^*$, with $h_0^* <
h_1^*$ for the thickness of the rescaled slab, at which transitions in
the qualitative behavior of the critical points of $\bar E_{00}$
occur. Basically, the result is that the uniform state is the unique
global minimizer whenever the effective thickness $h$ of the sample
satisfies $h < h_1^*$.  Furthermore, the uniform state is even the
unique critical point as long as $h < h_0^*$. The precise statement is
the following:
\begin{theorem} \label{thm-needle} %
  Let $h_0^* = \pi \sqrt{2}$ and let $h_1^* > h_0^*$ be the unique
  solution of the system of equations
\begin{align} \label{def-h1} %
F(\rho_m) = 0 && \text{ and } && G(\rho_m) = h_1^*,
\end{align}
where the functions $G$ and $F$ are defined in (\ref{eq:40}) and
(\ref{eq:41}), respectively. Depending on the value of $h$, we
then have
\begin{enumerate}
\item If $h < h_0^*$, then $\bar A = 0$ is the unique global
  minimizer of $\bar E_{00}$, and there are no other critical points
  of $\bar E_{00}$.
\item If $h_0^* \leq h < h_1^*$, then $\bar A = 0$ is the unique
  global minimizer of $\bar E_{00}$, but there exist non-trivial
  critical points of $\bar E_{00}$.
\item If $h = h_1^*$, then there are two global minimizers, given by
$\bar A = 0$ and by the unique positive solution of (\ref{eq:38})
 vanishing at the endpoints.
\item If $h > h_1^*$, then the unique positive solution of
  (\ref{eq:38}) vanishing at the endpoints is the unique minimizer of
  $\bar E_{00}$.
\end{enumerate}
Furthermore, the unique positive critical point $\bar A_h$ of $\bar
E_{00}$ obeys $\bar A_h(\xi) \sim \xi^{4/3}$ for $h = h_0^*$ and $\bar
A_h(\xi) \sim \xi$ for $h > h_0^*$.
\end{theorem}
Thus, the transition from trivial to non-trivial minimizers of $\bar
E_{00}$ occurs precisely when $h = h_1^*$. Numerically, the critical
values of the parameter $h$ are
\begin{align*} %
 h_0^* \ \approx \ 4.443, && h_1^* \ \approx \ 6.113.
\end{align*}
In terms of the original energy $\mathcal E$, the statement of Theorem
\ref{thm-needle} has the following interpretation.  In view of
\eqref{eq:14}, the critical values $h_0^*$ and $h_1^*$ of $h$ define
the respective critical values of $\bar \lambda$:
\begin{eqnarray*}
\bar\lambda_0^* = \left( { \gamma {h_0^*}^2 \over 3 } \right)^{\frac 13},
\qquad   \bar\lambda_1^* = \left( { \gamma  {h_1^*}^2 \over 3 }
\right)^{\frac 13}, 
\end{eqnarray*}
with the meaning that if one chooses $\lambda = \bar \lambda
\eps^{\frac 23} |\ln \eps|^{\frac 13}$, then the global minimizer of $\mathcal
E$ will be trivial when $\bar \lambda < \bar \lambda_1^*$, and
non-trivial when $\bar \lambda > \bar \lambda_1^*$, for fixed $\bar
\lambda$ and $\gamma$ as $\eps \to 0$. Similarly, the trivial
minimizer is expected to be the unique critical point of $\mathcal E$
for $\bar \lambda < \bar \lambda_0^*$ when $\eps \to 0$. Thus, the
transition to non-trivial minimizers for $\eps \ll 1$ is expected to
occur at $\lambda = \lambda_1^* \simeq \bar \lambda_1^* \eps^{\frac 23}
|\ln \eps|^{\frac 13}$. We would similarly expect the transition to
non-trivial minimizers to occur at this value of $\lambda_1^*$ in all
the energies: $E_0$, $E$, and $\mathcal E$. We note that $\lambda_1^*$
can be rigorously shown to give the asymptotic upper bound for the
critical value of $\lambda$ at which non-trivial minimizers emerge by
constructing suitable trial functions out of the non-trivial
minimizers of $\bar E_{00}$.  However, since the arguments in this
section are based on certain assumptions on the geometry of the
magnetic domains, the arguments are not rigorous in terms of a lower
bound for $\lambda_1^*$.

\subsection{Isolated needles}

In this section, we present a formal asymptotic derivation of $\bar
E_{00}$ in \eqref{eq:37}. At the onset of the transition from uniform
magnetization to a patterned state as the applied field is reduced it
seems natural to expect the appearance of thin slender needle-shaped
domains of magnetization opposing the applied field. Under this
assumption, it is possible to further reduce the energy $E_0$ to
obtain the precise information about the shape of these domains.

\medskip

The starting point of the analysis in this section is the reduced
energy $E_0$ in the form of \eqref{eq:E00}. We are interested in the
magnetization configuration in the form of a single needle. More
precisely, we assume that the configuration consists of a a single
needle-shaped domain in a sufficiently large sample (i.e. $\ell
\gtrsim 1$). In particular, $\mathrm{supp} \, \chi$ looks like the
characteristic function of a prolate ellipsoid of radius $r_0 \ll 1$,
extending across $\Omega$ in the direction of the easy axis. The
crucial observation for the analysis in this section is that due to
the slender geometry, the dominant stray field interaction is
restricted to slices normal to the easy axis. This interaction,
however, is logarithmic and hence does not see the precise shape of
the magnetic domains (a similar phenomenon occurs in a related model
\cite{m:cmp09}). In fact, for a domain pattern described above, to the
leading order in $r_0 \ll 1$ the Green's function $G_Q$ of the
operator $-\Delta_Q$ can be approximated by $G_0$ given by
\begin{align*}
  G_0(\mathbf r) = {|\ln r_0| \over 2 \pi} \, \delta (x_1).
\end{align*}
Indeed, $G_0$ gives the leading order behavior of the Green's function
for the operator $-\Delta_0 = -\Delta_\perp$.  Then the non-local term
in the definition of $E_0$ may be written as
\begin{align*}
  \int_{\mathbb R \times \TTT} \partial_1 \chi
  (-\Delta_Q^{-1}) \partial_1 \chi &\simeq \ {|\ln r_0| \over 2 \pi}
  \int_0^1 \int_\TTT \int_\TTT \partial_1 \chi(x_1, \mathbf r_\perp)
  \, \partial_1 \chi(x_1, \mathbf r_\perp') \, d^2 \mathbf
  r_\perp \, d^2 \mathbf r_\perp' \, dx_1 \notag \\
  &= \ {|\ln r_0| \over 2 \pi} \int_0^1 \left( \partial_1 \int_\TTT
    \chi(x_1, \mathbf r_\perp) d \mathbf r_\perp \right)^2 dx_1.
\end{align*}
This motivates to define the function $A : [0,1] \to \R$ by $A(x_1) =
\int_\TTT \chi(x_1, \cdot)$, denoting the cross-sectional area of the
needle in the slice at $x_1$. Then \eqref{eq:E00} turns into
\begin{align*} %
  E_0[\chi] \ \approx \lambda^2 \ell^2 + \int_0^1 \left( 2 \eps
    \int_\TTT |\nabla_\perp \chi| \, d \mathbf r_\perp - 4 \lambda A +
    {2 \gamma |\ln r_0| \over \pi} |A'|^2 \right) d x_1,
\end{align*}
where we again have used slenderness of the needle in the sense of
$\int_\Omega |\nabla \chi| \approx \int_\Omega |\nabla_\perp \chi|$.
Minimizing the interfacial contribution of the energy at each $x_1$
for fixed cross-sectional area then leads to the following expected
behaviour for minimizers:
\begin{align*}
 \inf_\chi E_0[\chi] \ \approx \ \lambda^2 \ell^2 + \inf_A E_{00}[A],
\end{align*}
where $E_{00}$ is given by
\begin{align} \label{eq:25} %
E_{00}[A] = \int_0^1 \biggl( 4 \eps \sqrt{\pi} \, A^{\frac 12} - 4
\lambda \, A + {2 \gamma |\ln r_0| \over \pi} |A'|^2 \biggr) d
x_1.
\end{align}
Note that according to \eqref{ratio}, at the transition, where
$\lambda \sim \gamma^{\frac 13} \eps^{\frac 23} |\ln \eps|^{\frac 13}$, we expect
$r_0 \sim \eps^{\frac 13} \gamma^{-\frac 13} |\ln \eps|^{-\frac 13}$ and, hence,
$|\ln r_0| \simeq |\ln (\gamma / \eps)^{\frac 13}|$.  Using these two
scalings and by introducing the rescaled variables
\begin{align*}
  \bar A = {\lambda^2 \over \eps^2} A, %
  && \xi = h \, x_1, %
  && \bar E_{00}[A] = {\lambda h \over 4 \eps^2} E_{00}[\bar A],
\end{align*}
we obtain the energy $\bar E_{00}$ in \eqref{eq:37}. Finally, since
$\bar A(0) > 0$ or $\bar A(h) > 0$ imply that there is a charge layer
at the plate's surface, causing a lot of stray field energy, to the
leading order the minimizers of $E_{00}$ are expected to satisfy
$\overline A(0) = \overline A(h) = 0$.

\subsection{Needle shapes and critical fields}

In this section, we investigate the minimizers of $\bar E_{00}$ and
give the proof of Theorem \ref{thm-needle}.  We first assume there
exists a positive critical point $\bar A_h(\xi) \in H^1_0((0, h))$ of
$\bar E_{00}$, i.e. $A_h > 0$ in $(0,h)$.  By standard ODE theory,
$\bar A_h \in C^2((0,h))$ satisfies the Euler-Lagrange equation for
\eqref{eq:37}, i.e.
\begin{align} \label{eq:38} %
  {1 \over \pi} {d^2 \bar A_h \over d \xi^2} = - 1 + {\sqrt{\pi} \over
   2 \sqrt{\bar A_h}} %
  && \text{ for $\xi \in (0,h)$}.
\end{align}
Note that \eqref{eq:38} admits a first integral
\begin{align} \label{eq:39} %
{1 \over 2 \pi} \left( {d \bar A_h \over d \xi} \right)^2 + \bar A_h
- \sqrt{\pi \bar A_h} = \pi C
 && \text{ for $\xi \in (0,h)$},
\end{align}
where $C \in \R$ is an arbitrary constant. Evaluating \eqref{eq:39} at
$\xi = 0$ and in view of $A_h(0) = 0$, we have $C \geq 0$.  Evaluating
\eqref{eq:39} at the maximum point of $\bar A_h$, yields $\max \bar
A_h \geq \pi$. We furthermore note that the solution of \eqref{eq:39}
is monotone in $(0,{h \over 2})$ and takes its maximum at $\xi = {h
  \over 2}$. It is convenient to introduce the rescaled needle radius
$\rho(\xi) \ = (\bar A_h(\xi) / \pi)^{\frac12}$ where in view of the
above the maximum $\rho_m = \rho(h/2)$ satisfies $\rho_m \geq 1$.
Integrating \eqref{eq:39} over $(0, {h \over 2})$, a straightforward
calculation yields
\begin{align} \label{dada} %
  \frac h2 \ = \ \sqrt{2} \int_0^{\rho_m} {\rho \, d \rho \over
    \sqrt{C + \rho - \rho^2}} \
  = \ \frac 1{\sqrt{2}} \left( \sec^{-1}(1 - 2 \rho_m) + 2
    \sqrt{\rho_m (\rho_m - 1)} \right).
\end{align}
Here we used the fact that $C = \rho_m ( \rho_m - 1)$, which follows
by evaluating \eqref{eq:39} at $\xi = h/2$.  In particular, for any
given $h > 0$ a positive critical point of $\bar E_{00}$ exists if and
only if $h = G(\rho_m)$ for some $\rho_m \geq 1$, where
\begin{align} \label{eq:40} %
  G(\rho_m) \ := \ \sqrt{2} \left( \sec^{-1}(1 - 2 \rho_m) + 2
    \sqrt{\rho_m (\rho_m - 1)} \right),
\end{align}
Differentiating \eqref{eq:40}, one gets that ${dG(\rho_m) \over d
  \rho_m} = {4 \sqrt{2 \rho_m (\rho_m - 1)} \over 2 \rho_m - 1} > 0$,
i.e, $G(\rho_m)$ is strictly monotonically increasing for $\rho_m \geq
1$, with $G(1) = \pi \sqrt{2}$ and $G(\rho_m) \to \infty$ as $\rho_m
\to \infty$. In particular, a non--trivial critical point of $\bar
E_{00}$ exists, if and only if $h \geq h_0^* := \pi \sqrt{2}$.  Since
for every local minimizer the Euler-Lagrange equation holds in every
interval of positivity, this also shows that for every $h < h_0^*$,
the only critical point of $\bar E_{00}$ in $H^1_0((0,h))$ is $\bar A
= 0$. This completes the proof about the existence or non--existence
of non--trivial critical points in (i) and (ii) in Theorem
\ref{thm-needle}.

\medskip

\begin{figure}
\begin{center}
   \includegraphics[height=3cm]{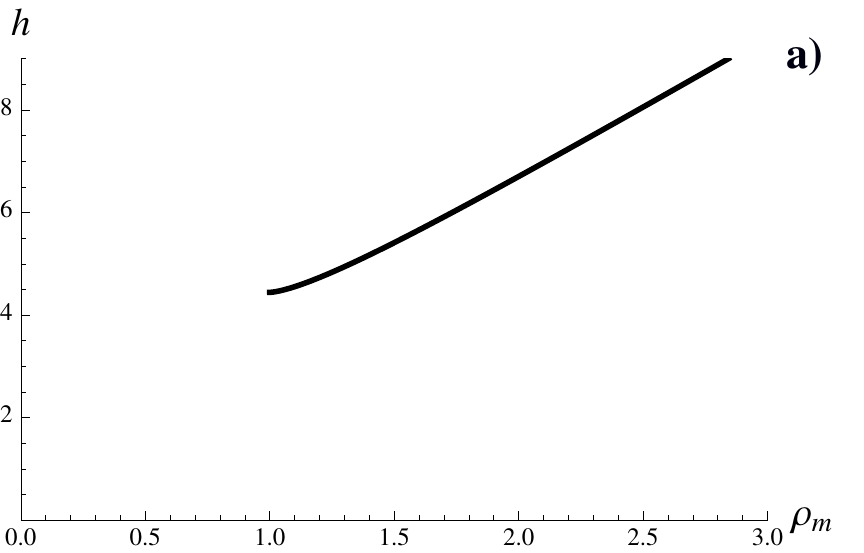} \hspace{8ex}
   \includegraphics[width=5cm]{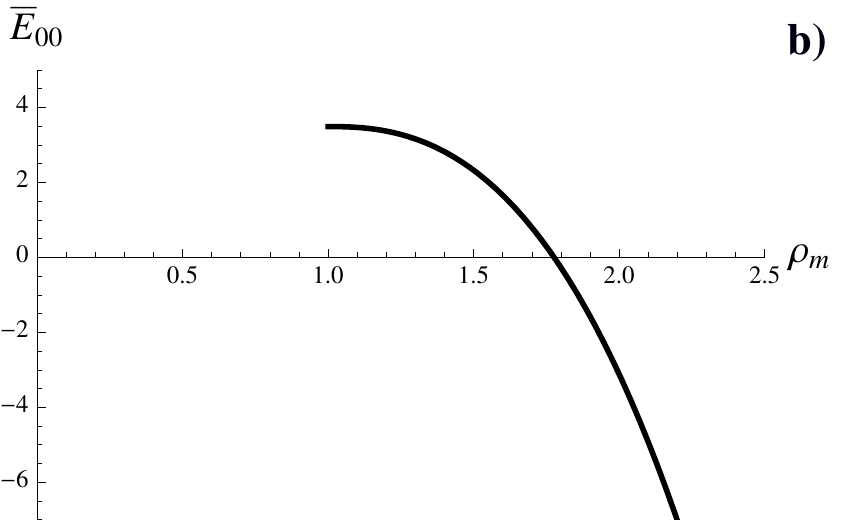}    
 \end{center}
 \caption{a) Parametric dependence of needle height on the radius,
   obtained from \eqref{eq:40}. b) Needle energy as a function of
   radius, obtained from \eqref{eq:41}.}
 \label{fig:h}
\end{figure}

Let us now consider the global minimizers of $\bar E_{00}$, which
exist, in view of coercivity and lower-semicontinuity of $\bar E_{00}$
for all $h > 0$. We first calculate the energy of the needle profile
calculated in the first part of the proof. For this, we define
\begin{align*}
  F(\rho_m) \ := \ \bar E_{00}[\bar A_h] \ = \ {1 \over \pi} \int_0^h
  \left( {d \bar A_h \over d \xi} \right)^2 d \xi - \pi h C.
\end{align*}
An explicit computation then yields
\begin{align} \label{eq:41} %
  F(\rho_m) \ = \ \frac \pi{\sqrt{18}} (3 - 4 \rho_m (\rho_m - 1))
  \sqrt{\rho_m (\rho_m - 1)} + \frac \pi{\sqrt 8} \sec^{-1} (1 - 2
  \rho_m).
\end{align}
Once again, differentiating this function with respect to $\rho_m$,
one gets ${dF(\rho_m) \over d \rho_m} = -{4 \pi \sqrt{2 \rho_m^3
    (\rho_m - 1)^3} \over 2 \rho_m - 1} < 0$, so that $F(\rho_m)$ is
strictly monotonically decreasing for $\rho_m \geq 1$, with $F(1) =
\pi^2/(2 \sqrt{2})$ and $F \to -\infty$ as $\xi \to \infty$. By
monotonicity of $G$ it then follows that $\bar E_{00}[\bar A_h]$ is
strictly decreasing in $h$ for $h \geq h_0^*$. In particular, for
$h_1^* > h_0^*$ defined by \eqref{def-h1}, we have $\bar E_{00}[\bar
A_h] > 0$ for $h \in (h_0^*, h_1^*)$ and $\bar E_{00}[\bar A_h] < 0$
for all $h \in (h_1^*, \infty)$ In particular, for $h \in (0, h_1^*)$,
the only global minimizer of $\bar E_{00}$ is given by $\bar A =
0$. The dependences of $h$ and $\bar E_{00}[\bar A_h]$ on $\rho_m$ are
shown in Fig. \ref{fig:h}.

\medskip

We now claim that for $h > h_1^*$ the global minimizer $\bar A$ is
unique and is given by $\bar A = \bar A_h$. Indeed, we first note that
$\bar A$ is not equal to $0$. In view of the above estimates on $\bar
E_{00}[A_h]$, every interval of positivity contains a point $x^*$ with
$\bar A(x^*) = \pi$. By strict monotonicity of $\bar E_{00}[A_h]$ as a
function of $h$, it is furthermore clear that $\xi = 0$ and $\xi = h$
are boundary points of the intervals of positivity of $\bar A$.
Suppose that $\bar A = 0$ on $I = [\xi_1, \xi_2] \subset\subset (0,h)$
where $\xi_1 \leq \xi_2$. By the above reasoning it follows that there
exist points $\xi_1' \in (0,\xi_1)$ and $\xi_2' \in (\xi_2, h)$ such
that $\bar A(\xi_1') = \bar A(\xi_2') = \pi$.  It follows that $\tilde
A$ defined by $\tilde A := \pi$ in $(\xi_1',\xi_2')$ and $\tilde A :=
\bar A$ outside of $(\xi_1',\xi_2')$ has lower energy than $\bar A$
contradicting the assumption that $\bar A$ is a minimizer.  This shows
that $\bar A > 0$ in $(0,h)$ and hence $\bar A = \bar A_h$.  This
completes the proof of (i)--(iv) in Theorem \ref{thm-needle}.

\medskip

In view of \eqref{eq:39} and \eqref{dada}, an explicit calculation
yields the following parametric equation for the needle profile for $0
< \xi < {h \over 2}$:
\begin{align*}
  &\xi \ = \ \frac{1}{\sqrt{2}} \left ( 2 \sqrt{(\rho_m-1)
      \rho_m}+\tan ^{-1}\left( (2 \rho -1) (2 \sqrt{(\rho_m-\rho )
        (\rho_m+\rho -1)})^{-1} \right)
  \right. \quad \notag \\
  &\qquad \left. -2 \sqrt{(\rho_m-\rho ) (\rho_m+\rho -1)} + \cot
    ^{-1}\left(2 \sqrt{(\rho_m-1) \rho_m}\right) \right),
\end{align*}
see Fig. \ref{fig:needle_hstar}. In particular, one easily checks that
for $0 < \xi \ll 1$ the behavior of the radius $\rho$ of the needle
near the tip is given by
\begin{align} \label{eq:43} %
 \rho(\xi) \ \sim \ \xi^{\frac 12} \quad \text{for } h > h^*_0 %
  && \text{and} && \rho(\xi) \ \sim \ \xi^{\frac 23} \quad \text{ for } h =
  h^*_0.
\end{align}
% Note that this behavior corresponds to the one we have used in our
% upper bound constructions in Section \ref{ss-upper}.
This concludes the proof of Theorem \ref{thm-needle}.
\begin{figure}
\centerline{\includegraphics[width=6cm]{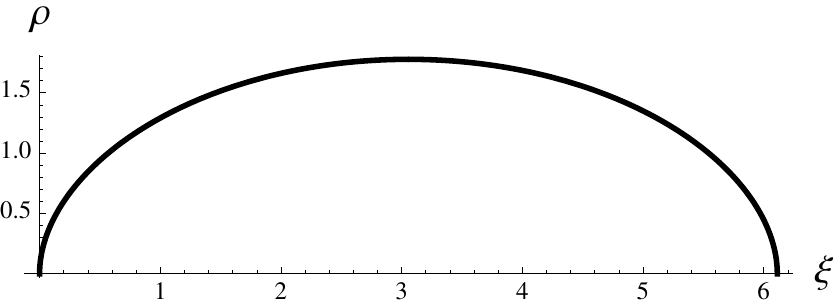}}
 \caption{The needle shape at $h = h_1^*$, obtained from \eqref{eq:38}.}
 \label{fig:needle_hstar}
\end{figure}

\section*{Acknowledgements} 
The authors are indebted to R. V. Kohn for pointing out similarities
of the considered problem with that arising from models of type-I
superconductors, as well as for many valuable discussions. The authors
would also like to acknowledge the hospitality of IAM at the
University of Bonn, where part of the work was done. C. B. M. was
supported, in part, by NSF via grants DMS-0718027 and DMS-0908279.

\bibliographystyle{plain}

\bibliography{branching}

\end{document}

%% file: def-branching.tex
\newtheorem{theorem}{Theorem}[section]

\newtheorem{proposition}[theorem]{Proposition}

\newtheorem{lemma}[theorem]{Lemma}

\newcommand{\eps}{\varepsilon}
\newcommand{\R}{\mathbb{R}}

\renewcommand{\AA}{{\mathcal A}}

\newcommand{\TTT}{\mathrm T}

\def\XXint#1#2#3{{\setbox0=\hbox{$#1{#2#3}{\int}$}
\vcenter{\hbox{$#2#3$}}\kern-.5\wd0}}

\newcommand{\upref}[2]{\hspace{-1.3ex}\stackrel{\eqref{#1}}{#2}\hspace{-0.6ex}}

\newcommand{\lupref}[2]{\hspace{0ex}\stackrel{\eqref{#1}}{#2}}

\newcommand{\lupupref}[3]{\hspace{0ex}\stackrel{\eqref{#1},\eqref{#2}}{#3}}

\newcommand{\mm}{{\mathrm {\mathbf m}}}
\newcommand{\hh}{{\mathrm {\mathbf h}}}
\newcommand{\uu}{{\mathrm {\mathbf u}}}
\newcommand{\vv}{{\mathrm {\mathbf v}}}
\newcommand{\ee}{\textbf e}
\newcommand{\nnu}{\boldsymbol \nu}

\newcommand{\OOO}{\mathcal O}
\newcommand{\LLL}{\mathcal L}
\newcommand{\EEE}{\mathcal E}
\newcommand{\SSS}{\mathcal S}

\newcommand{\MM}{\mathbf{M}}

\newcommand{\dist}{{\rm dist}\,} 
\renewcommand{\phi}{\varphi}

\newcommand{\nltL}[2]{\|#1\|_{L^2({#2})}}

\newcommand{\cciL}[1]{C_c^{\infty}(#1)}